\documentclass[10pt,a4paper]{article}
\usepackage{amsmath}  
\usepackage{amsfonts}
\usepackage{amssymb}
\usepackage{amsthm}
\usepackage{cite}
\usepackage{fullpage}
\usepackage[]{nicefrac}
\usepackage{bbm}
\usepackage{qcircuit}
\usepackage{authblk}
\usepackage{color}
\usepackage{float}
\usepackage[normalem]{ulem}
\usepackage{cancel}
\usepackage{graphicx}
\usepackage{varwidth}
\usepackage{tikz,tkz-graph}

\newcommand{\ignore}[1]{}
\newcommand{\ket}[1]{\left|#1\right\rangle}
\newcommand{\bra}[1]{\left\langle#1\right|}

\newcommand{\ketbra}[2]{ \left| #1 \right\rangle\left\langle #2 \right|}
\newcommand{\prnt}[1]{\left( #1 \right)}
\newcommand{\prntt}[1]{\left[ #1 \right]}\newcommand{\prnttt}[1]{\left\{ #1 \right\}}
\newcommand{\abs}[1]{\left| #1 \right|}

\newcommand\norm[1]{\left\lVert#1\right\rVert}

\newcommand\ceil[1]{\left\lceil #1 \right\rceil}
\DeclareMathOperator{\tr}{tr}

\def\th{\ensuremath{^\mathrm{th}} }
\def\BQP{\ensuremath{\textsc{BQP}}}
\def\PSPACE{\ensuremath{\textsc{PSPACE}}}
\def\PP{\ensuremath{\textsc{PP}}}

\newcommand{\Mod}[1]{\ \texttt{mod}\ #1}

\newtheorem{theorem}{Theorem}
\newtheorem{postulate}{Postulate}
\newtheorem{corollary}{Corollary}
\theoremstyle{definition}
\newtheorem{definition}{Definition}
\newtheorem{lemma}{Lemma}
\newtheorem{claim}{Claim}

\begin{document}

\bibliographystyle{unsrt}

\title{Fast-forwarding of Hamiltonians and Exponentially Precise Measurements} 

\author{Yosi Atia \& Dorit Aharonov} 
\affil{School of Computer Science and Engineering, 
The Hebrew University, Jerusalem,  Israel}
\maketitle
\begin{abstract}
In the early days of quantum mechanics, it was believed that the time 
energy uncertainty principle (TEUP) sets a bound on the 
efficiency of energy measurements, relating the duration, $\Delta t$, 
of the measurement, and the accuracy error in that measurement, $\Delta E$, by 
$\Delta E\cdot \Delta t\ge\frac{1}{2}$. 
Y. Aharonov and Bohm \cite{AB61} gave an example in which the principle 
doesn't hold; whereas Aharonov, Massar and Popescu \cite{AMP02} 
showed that under certain 
conditions the principle applies. Can we classify when and to 
what extent the TEUP is violated? 

Our main result is a 
theorem which unravels the source of such violations of the TEUP: 
such violations are equivalent to the ability to {\it fast 
forward} the associated Hamiltonian, 
namely, to simulate (using a quantum computer) its evolution for time $t$ using  significantly less than $t$ quantum gates. 
Precision measurements are thus in one to one correspondence with 
such fast-forwarding abilities, fundamentally linking quantum algorithms 
and precision measurements \cite{CPR00}. 
Our theorem is stated in terms of 
a modified TEUP, which we call the computational TEUP (cTEUP). 
In this principle the 
time duration ($\Delta t$) is replaced by the number of quantum 
gates required to perform the measurement, 
and we argue why this is 
more suitable to study if one is  
to understand the totality of physical resources required 
to perform an accurate measurement. 

The inspiration to our result is an intriguing example we provide: 
a family of Hamiltonians, based on Shor's 
algorithm \cite{Shor94}, which {\it exponentially} violates 
the cTEUP (as well as the TEUP), and which allow exponential fast forwarding.  
We further show that commuting local Hamiltonians, as well as any quadratic 
Hamiltonian of fermions on $n$ sites, can be fast forwarded. An  
important special case is Anderson localization. 
The work raises the question of finding a good physical criterion for 
fast forwarding; in particular, can   
many body localization systems be fast forwarded?  
We rule out a general fast-forwarding method for all 
physically realizable Hamiltonians (unless $\BQP=\PSPACE$). 
Connections to quantum metrology and to Susskind's  
complexification notion related to a wormhole's length are discussed, 
together with further directions.    
\end{abstract}

\section{Introduction}
In quantum mechanics, position and momentum are conjugate
variables, and it is well known that one can prove the position-momentum 
uncertainty principle $\Delta x \cdot \Delta p \ge \frac{1}{2}$ 
from the properties of the Fourier transform \cite{Peres93}.
Like position and momentum, frequency and time are also conjugate
variables, and since frequency represents energy in quantum theory, 
a natural question is whether time and energy also obey an 
uncertainty principle. 
Several obstacles arise immediately - time is not an operator but a scalar, hence it commutes with all quantum operators; the ``time'' quantity of a particle is not well defined; time doesn't evolve, and can't be a constant of motion. A common misconception of the time energy uncertainty principle 
(TEUP) is the following:\\
\textbf{TEUP Misconception} \textit{The duration $\Delta t$ of an energy measurement with accuracy error $\Delta E$ is bounded from below by }
\begin{equation}\label{eq:TEUPwrong}
\Delta E\cdot \Delta t \ge \frac{1}{2}.
\end{equation}

There are several other notions that are referred to in the literature
by the name time-energy uncertainty principle, or are sometimes 
discussed in its context 
(see, e.g., the 
survey of Busch \cite{Busch08}). In the related work section \ref{sec:related}  we will discuss our results 
in light of some of those (in particular, 
the Heisenberg limit\cite{GLM06,GLM11,ZpDK10,ZpDK12}
and the Mandelstam-Tamm relation
\cite{MT45}). 

The TEUP misconception, Equation \ref{eq:TEUPwrong} 
is indeed obeyed in many cases, such as measuring the energy of the electrons in an excited atom by detecting the emitted photon. The measurement duration depends on the lifetime of the excited state, and the lifetime is inversely proportional to the spectral linewidth which represents the uncertainty in energy (see, e.g., \cite{AMP02}).

However, already in 1961, Y. Aharonov and Bohm \cite{AB61} gave an example 
for a case in which the 
TEUP is violated: Suppose a Hamiltonian 
$H$ acts on a system $S$, and let $U_{meas.}$ be an energy measurement procedure which measures the energy 
of an eigenstate of the system $S$ with respect to $H$, to within some 
accuracy $\Delta E$. One can define a Hamiltonian $H_{meas.}$ which generates $U_{meas.}$:
\begin{equation} \label{eq:Hmeas}
U_{meas.}=e^{-iH_{meas.}\Delta t}. 
\end{equation}
Note that $H$ is the Hamiltonian  whose energy we are interested in measuring, can be viewed as the {\it input Hamiltonian}, while $H_{meas.}$ is an auxilary Hamiltonian which could be very different from $H$,  though must somehow be related to $H$ in order to perform an energy measurement of eigenstates of $H$. Aharonov and Bohm \cite{AB61} noticed that the measurement duration $\Delta t$ can be reduced arbitrarily 
if an amplified measurement Hamiltonian is used: 
\begin{equation} \label{eq:ABtrick}
U_{meas.}=e^{-iH_{meas.}'\Delta t'},
\end{equation}
with $H_{meas.}'=cH_{meas.}$ and $\Delta t'=\Delta t / c$, providing an 
arbitrarily large violation of the TEUP.  

On the other hand, Y. Aharonov, Massar and Popescu \cite{AMP02} showed (following initial results by Childs Preskill and Rene \cite{CPR00}) that
under the condition that {\it nothing is known about 
the Hamiltonian being measured}, such an amplification is impossible, and the 
TEUP does in fact hold.  
In other words, if the input Hamiltonian $H$ is treated as a {\it black box}, which the 
experimentalist can only turn on and off for some duration of time 
which he can control, but knows nothing about the inner makings of this box,
and in particular knows nothing about 
the eigenvalues and eigenvectors of the Hamiltonian, 
the TEUP holds. In fact, even lack of knowledge of only
the eigenvalues suffices
for the TEUP to hold. 

Thus, in some cases 
the TEUP is violated,  
whereas in the ``unknown Hamiltonian'' case 
it is obeyed. 
To what extent and in which situations 
can the TEUP be violated in physics? 

The main contribution of this paper is to connect this question to a notion 
which we call {\it fast forwarding} of Hamiltonians. 
Fast forwarding is the ability to simulate (using a quantum computer) 
the evolution of a given system governed by a certain Hamiltonian $H$ to within 
time $t$, but s.t. that the simulation takes time which is 
{\it much less than $t$}:  

\begin{definition} [fast forwarding a Hamiltonian] \label{def:FF} A normalized Hamiltonian $H$ ($\norm{H}=1$) acting on $n$ qubits can be 
$(T(n),\alpha(n))$-fast forwarded if for any $t\le T$, there exists a 
quantum circuit $\widetilde{U}$ with 
complexity $poly(n)$ which acts on the $n$ qubits and 
on additional $c=poly(n)$ ancilla qubits initialized to $0$,  
s.t.
\begin{equation}
\norm{(e^{-iHt}\otimes \mathbbm{1}_{2^c} -\widetilde{U})\ket{\psi}\otimes \ket{0}} \le \alpha 
\end{equation}
\end{definition}
 
It turns out that fast forwarding occurs if and only if the TEUP 
can be violated (a more precise definition of this violation will be given 
later). This suggests that the trade off between efficiency and accuracy in energy measurements can be studied using the computational language. The main object of this paper is to initiate such a study. 

We start in Section \ref{sec:cTEUP} where we introduce 
a variant of the TEUP which we call the computational TEUP (cTEUP), 
in which the time duration of the energy measurement
is replaced by the computational complexity of the 
energy measurement. We explain why we 
believe this to be a more adequate version to study in our context 
(and stress that this does not weaken our results, but just places them in 
a more accurate context).  
Section \ref{sec:ShorAsCounterexample} describes an intriguing example of 
an exponential violation of the cTEUP based on Shor's algorithm, which 
triggered this work. In this example one already sees the relation to 
fast forwarding.  
Section \ref{sec:FF} states and proves the 
equivalence between fast forwarding and exponentially precise efficient
energy measurements. We next proceed to study which Hamiltonians 
can be fast forwarded (or equivalently, their energy can be measured 
efficiently with very high accuracy). 
Section \ref{sec:FFPhysicalHamiltonians} shows that commuting local 
Hamiltonians,  
quadratic fermionic Hamiltonians, and Anderson localization are such systems. 
Section \ref{sec:2SparseFFimpBQPinPSPACE} proves that no general 
fast forwarding method exists for physical Hamiltonians, 
unless $\BQP=\PSPACE$.   
In Section \ref{sec:QAlgAndFF} we explore quantum
algorithms which can be associated with Hamiltonians, other than Shor's, 
and argue why none seem to exhibit 
fast forwarding (not even quadratic); we also remark on 
an observation raised by this study, related to the existence of a 
quantum algorithmic speed up for a cyclic version of the graph automorphism 
problem. 
Section \ref{sec:related}  
discusses the many connections between this work and other physical 
questions, including relations to metrology and sensing, the Heisenberg 
limit, and to a recent result by 
Susskind and Aaronson \cite{AS16p} regarding Susskind's conjecture on the relation between the complexity of certain quantum states and the length of 
non-traversible wormholes in quantum gravity\cite{Susskind16}. 
We end with some general conclusions and 
a few open questions in Section \ref{sec:conc}. 

We believe that this work makes an important step towards the fundamental question of characterizing and understanding 
the conditions for high resolution measurements;  
Our results suggest that studying this question 
within the framework 
of quantum information processing, and in particular  
taking into account 
the ability to apply quantum computations to aid the measurement 
process, might enable very interesting insights.
Hopefully, these insights will lead to experimental implications 
on precision measurements.

\section{A computational version of the TEUP} \label{sec:cTEUP}
Our first step is to suggest that a more adequate TEUP to investigate
is not the above postulate (Equation \ref{eq:TEUPwrong}) 
but a circuit complexity version of it.

To understand why such a modification is required, and what would that be, we return to the 
arbitrary violation of the TEUP given by \cite{AB61}. 
As Y. Aharonov and Bohm observed (see Equation 
\ref{eq:ABtrick}) 
the two resources of time and energy can be interchanged in their case, 
without affecting the accuracy of the measurement; 
increasing the norm of the interacting Hamiltonian can shorten the duration of the measurement. Since the TEUP 
doesn't take into account the interaction Hamiltonian's norm, 
there is no ``price'' for the large norm  
in the TEUP - we thus arrive at a violation.  
However, if one does take the norm into account in their case, no violation is achieved; 
see further in this section for more details.

What if we insist that the norm is kept bounded? 
It turns out that time can also be traded with some other resource.  
Consider the following trick. 
Let $U_{meas.}=e^{-iH_{meas.}t}$ be the time evolution corresponding to a Hamiltonian $H_{meas.}$ used by \cite{AB61} as the measurement Hamiltonian, whose
norm $\norm{H_{meas.}}$ can be arbitrary large. 
Writing 
\begin{equation}
U_{meas.}=\sum_j e^{-iE_j t} \ketbra{E_j}{E_j},
\end{equation}
as the time evolution by $H_{meas.}$, for time $t$, we notice that 
because the eigenvalues of $U_{meas.}$ are in the complex unit circle, 
we can always find 
an alternative Hamiltonian which generates the same
time evolution, but which is of bounded norm:   
$H'_{meas.}=\sum_j (E_j t \Mod 2\pi) \ketbra{E_j}{E_j}$.  
Evolving according to $H'_{meas.}$ for time $t'=1$ 
is equivalent to applying the unitary evolution $U_{meas.}$, 
and thus, we arrive at the same violation of the TEUP but with bounded 
norm and with finite time duration!  
Notice however, that now another resource is being heavily used: 
to apply $H'_{meas.}$, one needs to diagonalize the original 
Hamiltonian and compute its eigenvalues to extremely high precision.  
This could be highly demanding computationally. 

What is revealed by the above discussion is that once one allows 
manipulations while performing the energy measurement, such as 
increasing the norm, or modifying the Hamiltonian in other ways to 
achieve an ``equivalent'' energy measurement, the TEUP can be easily bypassed.  
Nevertheless the resources invested in the measurement have not decreased 
but were just interchanged with others.   
The ``correct'' notion that we would like to capture in the TEUP is not the 
time duration but the totality of
physical resources one is required to invest in a measurement.  
The underpinnings of the area of quantum computation 
(see \cite{BV97}) tell us exactly what is the right quantity 
to look at: the {\it computational 
complexity} of the measurement, 
namely, the size of the quantum circuit {\it simulating} 
the process of the measurement, where size is measured by 
the number of two-qubit quantum gates (see Definition in \cite{NC00}). 
This notion exactly takes into account 
all possible ways to apply a quantum process which results in a measurement 
of the energy of a state with respect to the given Hamiltonian.  
We thus postulate the following {\it computational complexity 
version} of the TEUP (which we denote by cTEUP):  

\begin{postulate}[computational TEUP (cTEUP)]    \label{conj:cTEUP} 
An energy measurement of an eigenstate of 
a Hamiltonian $H$, with accuracy error $\delta E$, satisfies
\begin{equation}
\delta E \cdot (\mathtt{measurement's~ computational~ complexity}) \in 
\Omega (1). 
\end{equation}
\end{postulate}

We note that $\delta E$ refers to  
the {\it accuracy error}, namely the difference 
between the correct eigenvalue and the outcome of the measurement. 

We need to be 
slightly careful in defining the error. First, the accuracy of the
outcome is only guaranteed with some probability, usually close to 
but not equal to $1$. We refer to this probability as the {\it confidence}.  
More precisely, we use the following notation:
 
\begin{definition} [$\eta$-accuracy]\label{def:eta-accuracy}
An energy measurement is said to have accuracy $\delta E$ with confidence
$\eta$ (we denote this as a measurement of $\eta$-accuracy $\delta E$) if 
given an eigenstate with energy $E$, the measurement outcome $E'$ satisfies

\begin{equation}\label{eq:eta}  
\Pr_{~~~E'} \prnt{\abs{E-E'}\le \delta E} \ge \eta. 
\end{equation}
\end{definition}

We usually set $\eta = 2/3$. 
Note that requiring the measurement to have accuracy error $\delta E$ with 
confidence $\eta$ is a slightly weaker requirement 
than the common requirement that the 
standard deviation is  $\delta E$. 
In particular, when the standard deviation is specified, 
it is assumed implicitly 
that the expectation of the outcome is the correct value $E$. 
However, the expectation of the outcome $E'$ of a
measurement of accuracy 
$\delta E$ and confidence $2/3$, 
might be arbitrarily far from $E$. Still, the {\it median} of 
many such measurements would be within 
$\delta E$ from $E$ with probability which approaches $1$   
exponentially fast in the number of repetitions 
(see the Confidence Amplification lemma, 
Lemma \ref{lem:ConfAmp}). The other direction does hold.  
When the 
expectation of the measurement is the correct value $E$, 
its standard deviation $std(E)$ can be seen to provide an upper  
bound on the $\nicefrac{2}{3}$-accuracy $\delta E$ (a la Definition 
\ref{def:eta-accuracy}): 
$\delta E \le \sqrt{3} std(E)$\footnote{This holds only if we assume that the values of the measurements 
are distributed by some smooth continuous distribution. 
If they are not, a slight modification of the definition is required: 
One needs to define $\delta E^\uparrow$ to be the supremum over all 
values for which $ \Pr_{E'} \prnt{\abs{E-E'}< \delta E^\uparrow} < \eta$ holds. 
Let $E$ be the energy of the Hamiltonian's eigenstate $\psi_E$. Assuming that the expectation of the energy measurement $\bar{E}$ coincides with $E$, 
\begin{equation*}
std(E)= \sqrt{\sum_{E'} \Pr(E') (E'-\bar{E})^2} \ge \sqrt{ \sum_{E': \abs{E'-\bar{E}}> \delta E^\uparrow} \Pr(E') (E'-\bar{E})^2} \ge \sqrt{ \sum_{E': \abs{E'-\bar{E}}> \delta E^\uparrow} \Pr(E') (\delta E^\uparrow) ^2}
\ge \sqrt{\frac{(\delta E^\uparrow)^2}{3}}=\frac{\delta E^\uparrow}{\sqrt{3}}
\end{equation*}
where in the last inequality we have assumed that 
$\delta E^\uparrow$ is the $2/3$ accuracy, with the modified definition. 
}.

There is one other source of error in the measurement process: 
the amount by which the measured 
state is modified, or {\it demolished}. Often 
one is interested in leaving the measured 
state in tact, as in non-demolition measurements \cite{BVT80}. 
To quantify this one can use any reasonable metric on quantum states, 
e.g. the fidelity or the trace metric \cite{NC00}. 
For now, we assume by default that the demolition of the state 
is set to be polynomially small in the number of qubits in the system.   
As we will see later, if the demolition error is that small, the 
confidence parameter can be easily amplified and thus the exact choice of 
$\eta$ doesn't matter to the question of whether a violation of the cTEUP 
is possible and to what extent\footnote{We have not investigated the case in which 
large demolition error is allowed; it might be that this 
does enable considerably more efficient measurements.}. 
 
Armed with this modified principle, Postulate \ref{conj:cTEUP}, 
we are able to clarify more adequately 
the true physical resources required for highly 
accurate and efficient measurements.
We stress that using the cTEUP instead of the TEUP 
does not make the question easier but just more 
fundamental, as it is closer to what we really want to understand, 
which is what are the physical limitations on precision measurements. 
Notice that violating the cTEUP 
implies violating the TEUP by definition; 
the other way round is not true, as in \cite{AB61} (see Section 2.1). 

\subsection{Completely known Hamiltonians} \label{sec:CompletelyKnownH}
A first natural question is: 
perhaps when taking all the resources into account, 
the cTEUP does hold? 

Indeed, the example in \cite{AB61} is no longer a counterexample. 
This is because 
the increase in the norm is 
reflected also in an increased computational complexity of the 
measurement:  
Let $f(n)$ be the time 
complexity of simulating the 
Hamiltonian $H_{meas.}$ of Equation \ref{eq:Hmeas} for one time unit. 
The naive way to simulate $cH_{meas.}$ for one time unit, in order 
to improve the accuracy by a factor of $c$,
is to concatenate $c$ copies of the circuit implementing $e^{-iH_{meas.}}$. 
This yields a total  time complexity $cf(n)$ - 
and the factor $c$ cancels with the 
one we get for the improvement in accuracy. Thus the cTEUP holds 
in the \cite{AB61} example.     

However, it turns out that a simple counterexample to the
cTEUP does exist. 
Here is a way to achieve an  
infinite violation of the cTEUP (as well as of the TEUP) 
using a simple Hamiltonian on $n$ spins (or qubits). Let
\begin{equation} \label{eq:hsigma}
H=\sum_{i=0}^{n}\sigma^z_i. 
\end{equation}
Given an eigenstate, which is a tensor product of the 
eigenstates of each of the $\sigma^z$'s, 
a measurement of each of the spins in the eigenbasis of the Pauli $\sigma^z$, 
(the computational complexity of this measurement is $O(n)$) 
reveals the eigenvalue to infinite precision, namely, with
$\delta E=0$. 
The demolition error however might be very large since most eigenstates 
are superposition of computational basis states. 
To avoid demolition altogether, 
an alternative measurement can be performed efficiently, using  
standard quantum computation tricks: 
Add a register of $\log(n)$ qubits all initiated in the state $0$, and 
apply the unitary version of the classical computation 
which computes $w(i)$, the number 
of 1's in the string $i$ of the original system, and writes it down on 
the additional register. In other words, apply the unitary operator:  
\begin{equation} 
U|i\rangle|0^{\log n}\rangle=|i\rangle|w(i)\rangle
\end{equation}
this can be done using $n$ times $poly(\log n)$ gates \cite{NC00}. 
Now measure the right register, 
which gives the correct energy with $\delta E=0$.

More generally, consider the $n$ qubit Hamiltonian 
$H = \sum_i \lambda_i \ket{\psi_i}\bra{\psi_i}$, and assume that we 
have full knowledge of its eigenstates and eigenvalues in the following 
sense: the functions
$\ket{i}\mapsto \ket{\psi_i}$ and $i\mapsto \lambda_i$ can be 
computed by a quantum computer in polynomial time in $n$.  
An infinite violation of the cTEUP can be achieved:  
One can first apply the unitary $U=\sum_i\ketbra{i}{\psi_i}$ on the state to be measured, use the function $i \mapsto \lambda_i$ to write the energy on an ancilla register, and measure the ancilla. Finally apply $U^{-1}$ to derive the original state again without any deviation. 

These infinite violations assume full knowledge of the 
eigenstates and eigenvalues of the Hamiltonian in the above sense.

\subsection{Hamiltonians with unknown Eigenvalues} 
It turns out, 
that when nothing is known about the 
Hamiltonian, or even just about its eigenvalues,
the cTEUP holds. The results of \cite{AMP02} show this for completely unknown Hamiltonians, namely, Hamiltonians which are given as black boxes. In fact, the proof applies as is also for Hamiltonians whose eigenstates are known but their eigenvalues are not. Theorem \ref{thm:TEUP23confidence} in
Appendix \ref{apndx:TEUPforUnknown} provides the exact statement and proof of 
the TEUP for Hamiltonians with unknown eigenvalues\cite{AMP02}, for 
completeness; it is slightly adapted to work 
in our terminology of Definition \ref{def:eta-accuracy} 
rather than in the mean deviation terminology of 
\cite{AMP02}. It is straight forward to argue that this theorem 
implies that also the cTEUP holds for Hamiltonians with unknown eigenvalues.  

\begin{theorem} [cTEUP for unknown Hamiltonians] \label{thm:cTEUPUnknown}
Let $H$ be a Hamiltonian whose eigenvalues are unknown (namely,  
 $H$ is taken from a set of Hamiltonians all of which 
have the same set of 
eigenvectors, but we know nothing about their eigenvalues). 
Let $G(H)$ be a quantum circuit which
applies the unitary $e^{-iH}$, given as a black box. 
Let $C(n)$ denote the computational complexity of a quantum circuit which, 
when given an input eigenstate of $H$, and which has 
access to $H$ only through the black 
box $G(H)$, performs
an energy measurement of the input state with respect to $H$ 
with accuracy $\delta E$ and confidence $2/3$. Then $C(n)$ satisfies:
\begin{equation}
\delta E \cdot C(n) \in \Omega(1).
\end{equation}
\end{theorem} 
\begin{proof} The proof follows trivially from Theorem 
\ref{thm:TEUP23confidence}, which
gives a lower bound for the total time duration $\Delta t$ that the Hamiltonian
$H$ is applied to achieve such an accurate energy measurement, while taking into account the possibility of applying H in intervals and on different parts of the system.   
Since the access to $H$ is only by using the circuit $G(H)$, which 
applied $H$ for one time unit, the number of 
instances of $G(H)$ being used is $\Omega(\Delta t)$. 
Hence $C(n)\in \Omega (\Delta t)=\Omega(1/ \delta E)$ by Theorem \ref{thm:TEUP23confidence}.  
\end{proof}

Given the two extreme cases described in the above two 
subsections, of completely known Hamiltonians versus 
Hamiltonians where no information on the eigenvalues is given, 
we are confronted with the question: 
which Hamiltonians of the more common type, 
that are neither fully known nor fully unknown, 
allow violating the cTEUP? And to what extent? 

\section{An Exponential violation of the cTEUP based on Shor's algorithm} \label{sec:ShorAsCounterexample}
We describe an enlightening example, based on Shor's 
 polynomial time quantum algorithm \cite{Shor94} 
 for finding the prime factors of a
given integer, which exhibits an exponential speed-up over all known classical 
algorithms for the same task. This algorithm
gives rise to a family of Hamiltonians 
whose eigenvalues and eigenvectors are {\it not} known to us in advance,
(computationally) and yet these Hamiltonians constitute 
a counterexample to the cTEUP (postulate \ref{conj:cTEUP}); moreover, 
they violate it {\it exponentially}.

To define the Hamiltonians, we first recall the essential steps in the
algorithm. It factors 
an $n$-bit number $N$ by finding the order $r$ 
of a randomly chosen $y$ co-prime to $N$ (i.e., $gcd(y,N)=1$), namely  
the period of the sequence 
$y^0, y^1, y^2...$. (If one knows how to find such an $r$, 
the factors of $N$ can then be found 
using an efficient classical algorithmic procedure \cite{NC00}).   
To find the order $r$, 
the algorithm uses the unitary $U_{N,y}$ acting on $n$ bit strings, 
defined as the application of multiplication by $y$ 
modulo $N$, where as usual we identify numbers with their
binary representation, namely strings of $n$ bits: 

\begin{equation}
U_{N,y}\ket{x}=\begin{cases}
\ket{x\cdot y ~\mathtt{mod}~ N} & 0\le x<N \\
\ket{x} & otherwise\\
\end{cases}
\end{equation}

Recall the orbit-stabilizer theorem  \cite {Rot12} which implies 
that $U_{N,y}$ partitions the 
set $\prnttt{0,1,...N-1}$ into orbits, each one being 
the orbit of some representative element in the set, 
and that the size of each orbit divides $r$. 

Our candidate Hamiltonian for violating the cTEUP is 
 defined as  
\begin{equation} \label{eq:ShorsNiceH}
H_{N,y}=U_{N,y}+U_{N,y}^\dagger.
\end{equation}

$H_{N,y}$ has only two non-zero elements in each row. 
Note that in physics, $H_{N,y}$ describes a tight binding model of 
several disjoint 1D lattices with periodical 
boundary conditions; except here these lattices are not physical but virtual, 
and in particular the size of these lattices
is exponential in the physical system's size $n$. 
We also note that $H_{N,y}$ can be simulable efficiently by a quantum computer,
in time $poly(n)$, because it has only two non-zero
 efficiently-computable elements in a row; 
that such Hamiltonians can be implemented efficiently is well known 
\cite{AT03}. 

To achieve efficient and exponentially accurate 
measurement of the eigenvalues of $H_{N,y}$, we use the fact
that $H_{N,y}$ shares the same 
eigenvectors with $U_{N,y}$, and their eigenvalues are related in a simple way.  
To calculate the eigenvalues of $U_{N,y}$, denote by
$x_\ell$ a representative of the $\ell$\th orbit in the partitioning of 
$\prnttt{0,1,...N-1}$ into orbits by $U_{N,y}$; denote this $\ell$th orbit by  
$\mathcal{O}(x_\ell)$. Then 
the eigenstates of $U_{N,y}$ (and of $H_{N,y}$) are of the form
\begin{equation}
\ket{\psi_{\ell,k_\ell}}=\sum_{j=0}^{\abs{\mathcal{O}(x_\ell)}-1} e^{\frac{2\pi i j k_\ell}{\abs{\mathcal{O}(x_\ell)}}} \ket{x_\ell \cdot y^j \Mod N}~~~~~~~~~\begin{array}{c}
\ell\in\prnttt{1,2,...,\# orbits}\\
k_\ell\in \prnttt{0,1,...,\abs{\mathcal{O}(x_\ell)}-1}\\
\end{array} 
\end{equation}
The eigenvalue of $\psi_{\ell,k_\ell}$ with respect to $U_{N,y}$ is 
$e^{i\varphi}=e^{\frac{2\pi i k_\ell}{\abs{\mathcal{O}(x_\ell)}}}$.
The eigenvalue with respect to $H_{N,y}$ is $E_{\ell,k_\ell}=2\cos(\varphi)$.

The efficient estimation of the eigenvalues of $H_{N,y}$ to within 
exponential accuracy is a standard exercise 
in quantum computation, by applying phase estimation with respect to the 
unitary $U_{N,y}$, and using the fact that exponential powers of $U_{N,y}$
can be applied efficiently using modular exponentiation. 
This is shown more rigorously in the following theorem. 
This proves that these Hamiltonians exhibit an 
exponential violation of the cTEUP, postulate \ref{conj:cTEUP}. 

\begin{theorem} \label{thm:ShorCounterexample}
Consider $H_{N,y}$ as above, such that $gcd(y,N)=1$, and where $N$ is an 
$n$-bit integer. 
There exists an energy measurement procedure which 
given any eigenstate of $H_{N,y}$ has 
$\nicefrac{2}{3}$-accuracy $\delta E$ such that: 
\begin{equation} 
\delta E \cdot (\mathtt{measurement's~ computational~ complexity}) 
\approx 1/exp(n).
\end{equation}
The measurement procedure is such that
the given eigenstate remains 
in tact.  
\end{theorem}

Importantly, we notice that though it might seem that the eigenvectors and 
eigenvalues of the Hamiltonians $H_{N,y}$ 
are known here in advance, in fact they are not - 
because they depend on $r$, which is computationally 
not known, and so this violation occurs despite the fact that we are 
not in the situation of a fully known Hamiltonian case.  To prove Theorem \ref{thm:ShorCounterexample}, we could have used the standard Fourier transform based phase estimation procedure. We choose to avoid this here and use Kitaev's original phase estimation procedure\cite{Kitaev95}, with classical post processing, since this emphasizes the role of fast forwarding in the process. 

\begin{proof}
Let  $\psi_\varphi$ be an eigenstate of $U_{N,y}$. By adding a control 
qubit in the state $\frac{1}{\sqrt{2}}(\ket{0}+\ket{1})$ in another 
register, and applying $U_{N,y}^t$ conditioned on the control 
qubit being $1$, the state of the control qubit becomes: 
\begin{equation} 
\ket{a_t}=\frac{1}{\sqrt{2}}(\ket{0}+e^{-i\varphi t }\ket{1}).
\end{equation} 
Measuring this qubit  
in the $\{\ket{+},\ket{-}\}$ basis yields ``$-$'' with 
probability $p=\sin^2(\varphi t/2)$.

Roughly, the idea is this.  
We start with $t=1$, prepare polynomially many such control qubits, 
measure all of them in the $\{\ket{+},\ket{-}\}$ basis, 
and by taking the majority we can get a 
very good estimate of one bit of $p$, and thus of $\varphi$ and $E$.
Chernoff bound tells us that by using polynomially many control qubits, 
we should have exponentially good confidence of this bit.  
To gain information about further bits, we can then repeat for
$t=2,4,8,...$, multiplying the time by two for every extra 
bit we want to estimate. We note that 
modular exponentiation allows an efficient implementation of $U_{N,y}^t$ 
for exponential values of $t$, so this is allowed. 
However there is some subtlety here. 
When the value which we are trying to estimate 
lies exactly or extremely close to the borderline between two possibilities,
e.g., $p$ is exponentially close to $1/2$, 
we are likely to make an error; the Chernoff bound only gives polynomial 
error with exponential confidence.  

We overcome this technical issue by a simple trick: we will indeed cut the size of the 
interval containing $\varphi$ to half every time we want to increase the accuracy, 
but the interval will not be either the left or the right side of the current 
interval - it can also be in the middle.

Formally, define 
$t_j= 2^{j-1}$ 
for $j\in \{1,2,...poly(n)\}$.  
In the $j$\th iteration the phase measured is assumed to be in the interval $[\varphi_{min}^{j}, \varphi_{max}^{j}]$, with $\Delta_j \triangleq \varphi_{max}^{j}-\varphi_{min}^{j}=\pi/t_j$. We generate $m=poly(n)$ many control qubits in 
the state 
\begin{equation} \label{eq:biasedcoin}
\frac{1}{\sqrt{2}}(\ket{0}+e^{-i (\varphi-\varphi_{min}^{j})\cdot t_j}\ket{1}),
\end{equation} 
apply Hadamard on each one of these qubits 
and measure in the $\{\ket{0},\ket{1}\}$ basis. 
This can be done efficiently:  the 
phase $\varphi t_j$ is added efficiently by modular exponentiation, and the known phase dependent on $\varphi_{min}^{j}$ is added by using standard quantum computation techniques 
(See Figure \ref{fig:KitaevPE2} for a schematic description of the circuit for 
one such control qubit). The outcome $1$ is achieved
with probability 
\begin{equation}  \label{eq:PvsE}
p_j=\sin^2 ((\varphi-\varphi_{min}^{j})t_j/2)=
\sin^2 \prnt{\frac{\varphi-\varphi_{min}^{j}}{\Delta_j} \cdot \frac{\pi}{2}}
\end{equation}
Let $\widetilde p_j$ be the fraction of the measurements with $1$ outcome, and let $\widetilde \varphi_j= \frac{2}{t_j}\arcsin(\sqrt {\widetilde {p}_{j}})+\varphi_{min}^j$.
It will serve as the current estimate of $\varphi$.  
For the $(j+1)$th iteration's interval (with length $\Delta_{j+1}$), 
we choose
\begin{equation} \label{eq:Emin}
\varphi_{min}^{j+1}=\begin{cases}
\varphi_{min}^{j} & ~~~\widetilde \varphi - \varphi_{min}^{j} \le \frac{\Delta_{j}}{3} ~~~ \prnt{\widetilde{p} \le \frac{1}{4}}\\
\varphi_{min}^{j} + \frac{\Delta_{j}}{4} & ~~~ \frac{\Delta_{j}}{3} < \widetilde \varphi - \varphi_{min}^{j} \le \frac{2\Delta_{j}}{3} ~~~ \prnt{\frac{1}{4} < \widetilde p \le \frac{3}{4}} \\
\varphi_{min}^{j} + \frac{\Delta_{j}}{2} & ~~~ \widetilde \varphi - \varphi_{min}^{j} > \frac{2\Delta_{j}}{3} ~~~ \prnt{\widetilde p > \frac{3}{4}}\\
\end{cases} 
\end{equation}
and of course 
$\varphi_{max}^{j+1} = \varphi_{min}^{j+1}+\Delta_{j+1}$.

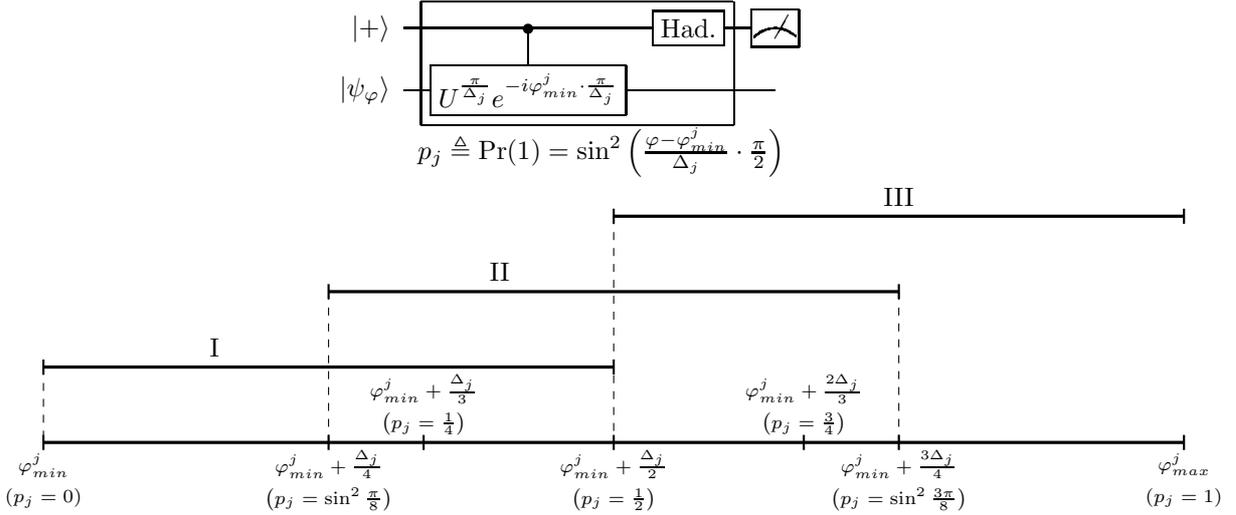
\begin{figure} [H]
\begin{center}
$
\Qcircuit @C=1em @R=.7em {
\lstick{\ket{+}} & \ctrl{1} & \gate{\mathrm{Had.}} &\meter
\\ 
\lstick{\ket{\psi_\varphi}}& \gate{U^{\frac{\pi}{\Delta_j}} e^{-i\varphi_{min}^j\cdot \frac{\pi}{\Delta_j}} } &\qw &\qw
\gategroup{1}{2}{2}{3}{.7em}{-}
}
$

$
p_j\triangleq \Pr(1)=\sin^2\prnt{\frac{\varphi-\varphi^j_{min}}{\Delta_j}\cdot \frac{\pi}{2}}
$
\\
\begin{tikzpicture}
\draw [very thick] (0,3) -- (15,3);
\draw [thick] (0,2.9) -- (0,3.1) node [pos=-2,align=center] {\scriptsize {$\varphi_{min}^j$} \\  {\scriptsize $\prnt{p_j=0}$}};
\draw [thick] (15/4,2.9) -- (15/4,3.1) node [pos=-2,align=center] { {\scriptsize $\varphi_{min}^j+\frac{\Delta_j}{4}$}\\ {\scriptsize $\prnt{p_j=\sin^2\frac{\pi}{8}}$}};
\draw [thick] (15/3,2.9) -- (15/3,3.1) node [pos=+3,align=center] {\scriptsize {$\varphi_{min}^j+\frac{\Delta_j}{3}$}\\ {\scriptsize $\prnt{p_j=\frac{1}{4}}$}};
\draw [thick] (15/2,2.9) -- (15/2,3.1) node [pos=-2,align=center] {\scriptsize {$\varphi_{min}^j+\frac{\Delta_j}{2}$}\\ {\scriptsize $\prnt{p_j=\frac{1}{2}}$}};
\draw [thick] (15*2/3,2.9) -- (15*2/3,3.1) node [pos=+3,align=center] {\scriptsize {$\varphi_{min}^j+\frac{2\Delta_j}{3}$}\\ {\scriptsize $\prnt{p_j=\frac{3}{4}}$}};
\draw [thick] (15*3/4,2.9) -- (15*3/4,3.1) node [pos=-2,align=center] {\scriptsize {$\varphi_{min}^j+\frac{3\Delta_j}{4}$}\\ {\scriptsize $\prnt{p_j= \sin^2\frac{3\pi}{8}}$}};
\draw [thick] (15,2.9) -- (15,3.1) node [pos=-2,align=center] {\scriptsize {$\varphi_{max}^j$}\\ {\scriptsize $\prnt{p_j=1}$}};
\draw [very thick] (0,4) -- (7.5,4) node [pos=0.3, above] {I};
\draw [thick] (0,3.9) -- (0,4.1);
\draw [thick] (7.5,3.9) -- (7.5,4.1);

\draw [very thick] (15/4,5) -- (15*3/4,5) node [pos=0.3, above] {II};
\draw [thick] (15/4,5.1) -- (15/4,4.9);
\draw [thick] (15*3/4,5.1) -- (15*3/4,4.9);

\draw [very thick] (15/2,6) -- (15,6) node [pos=0.5, above] {III};
\draw [thick] (15,5.9) -- (15,6.1);
\draw [thick] (15/2,5.9) -- (15/2,6.1);

\draw [dashed] (0,4)--(0,3);
\draw [dashed] (15/4,5)--(15/4,3);
\draw [dashed] (15/2,6)--(15/2,3);
\draw [dashed] (15*3/4,5)--(15*3/4,3);

\end{tikzpicture}

\end{center}
\caption{\label{fig:KitaevPE2}Phase estimation without Fourier transform. 
At the $j\th$ iteration of the procedure, 
the phase $\varphi$ is assumed to be in some interval $[\varphi_{min}^j,\varphi_{max}^j]$ 
of length $\Delta_j=2^{1-j}\pi$, which is cut to half at every iteration. 
At each iteration, the circuit (top) is applied $m=poly(n)$ 
times, and $\widetilde p_j$ denotes the ratio of $1$s measured. 
The next interval is chosen as part I (bottom of figure) 
if $\widetilde p_j\le\frac{1}{4}$, part III if $ \widetilde p_j>\frac{3}{4}$, and 
part II otherwise. 
$\varphi$ is outside the new interval only if $|p_j - \widetilde p_j|\ge\frac{1}{4}-\sin^2\prnt{\frac{\pi}{8}}$, and
by Chernoff, this probability is exponentially small in $m$. 
After $\ell=poly(n)$ iterations the interval's size
is $2^{-\ell}\Delta_1$, and by the union bound the confidence is 
exponentially close to 1.}
\end{figure} 

Assuming we chose the correct interval in every iteration, after $\ell$ 
iterations we know the phase with accuracy $\Delta_{\ell+1}= 2^{-\ell}\pi$.

\begin{lemma} \label{lem:TernaryKitaev}
Let $e^{i\varphi}$ be the eignenvalue of an $n$ qubit unitary $U$, and   $\varphi\in [\varphi_{min}^1,\varphi_{min}^1+\pi]$ for a known $\varphi_{min}^1$. The phase estimation procedure described above, denoted $f(U,[\varphi_{min}^1,\varphi_{min}^1+\pi])$, finds $\varphi$ with accuracy $2^{-\ell}\pi$ and with confidence greater than $1-\ell e^{-m/160}$, where $\ell$ is the number of iterations, and $m$ is the number of times the circuit in Figure \ref{fig:KitaevPE2} is applied in each iteration.
\end{lemma}
\begin{proof}
In each iteration there are four scenarios where we choose the wrong interval for the next iteration. 
We use the Chernoff bound and its monotonicity to find a bound on the error probability in each scenario, and then pick the largest one. 
\begin{gather}
\Pr(\widetilde{p}_{j}\ge \gamma ) \le e^{-\frac{mp_{j}}{3}(\frac{\gamma}{p_{j}} -1)^2 } ~~~~ \gamma > p_{j} \\ \label{eq:worstChernoff}
\Pr(\widetilde{p}_{j}\le \gamma ) \le e^{-\frac{mp_{j}}{2}(\frac{\gamma}{p_{j}} -1)^2 } ~~~~ \gamma < p_{j}
\end{gather}

\begin{enumerate}
\item The phase is in I$-$II (see figure \ref{fig:KitaevPE2}), i.e., $\varphi\in [\varphi_{min}^{j},\varphi_{min}^{j}+\frac{\Delta _{j}}{4})$, but we measured $\widetilde \varphi_j > \varphi_{min}^{j}+\frac{\Delta_{j}}{3}$ thus picked intervals II or III: 
\begin{equation}
\Pr\prntt{\widetilde{p}_{j}>\frac{1}{4} \cap p_{j}<\sin^2(\frac{\pi}{8})} \le e^{-(\frac{1}{4\sin^2(\frac{\pi}{8})}-1)^2 m \sin^2(\frac{\pi}{8})/3}=e^{-m(\sqrt 2 -1)/12\sqrt 2}
\end{equation}
(we used $\sin^2(\pi/8)=\frac{\sqrt{2}-1}{2\sqrt{2}}$)
\item The case symmetric to 1; $\varphi$ is in III$-$II and intervals I or II were chosen: 

\begin{equation}
\Pr\prntt{\widetilde{p}_{j}<\frac{3}{4} \cap p_{j}>\sin^2(\frac{3\pi}{8})} \le e^{-(\frac{3}{4\sin^2(\frac{3\pi}{8})}-1)^2 m \sin^2(\frac{3\pi}{8})/2}=e^{-m \frac{(1+\sqrt{2})^3}{8\sqrt{2}}}
\end{equation}
\item The phase is in (I$\cap$II)$-$III, i.e.,  $\varphi\in [\varphi_{min}^{j}+\frac{\Delta _{j}}{4}, \varphi_{min}^{j}+\frac{\Delta_{j}}{2})$ but we measured $\widetilde \varphi > \varphi_{min}^{j}+\frac{2\Delta _{j}}{3}$ thus picked interval III:
\begin{equation}
\Pr\prntt{\widetilde{p}_{j} > \frac{3}{4} \cap  p_{j} \in \Big[\sin^2(\frac{\pi}{8}), \frac{1}{2}\Big)} \le e^{-\frac{m}{24}}
\end{equation}
\item The case symmetric to 3; $\varphi$ is in (III$\cap$II$)-$I, and interval I was chosen:
\begin{equation}
\Pr\prntt{\widetilde{p}_{j} < \frac{1}{4} \cap p_{j} \in \Big(\frac{1}{2}, \sin^2(\frac{3\pi}{8})\Big] } \le e^{-\frac{m}{16}}
\end{equation}
\end{enumerate}
All these errors are smaller than $e^{-m/160}$. The probability to be in the wrong interval after $\ell$ iterations  follows from the union bound and is at most $\ell e^{-\frac{m}{160}}$.
\end{proof}

In general,we have no prior knowledge that $\varphi$ is in some window of size $\pi$, and so it can be any value in $[0,2\pi]$. Using $f(U,0)$ on an unbounded $\varphi$ creates an ambiguity: $\varphi$ and its mirror value $2\pi-\varphi$ would create the same distribution of measurement outcomes in every iteration. 
The following algorithm solves the ambiguity by trying to run $f$ on a range containing both values:\\
{\\}
\noindent\fbox{%
\begin{varwidth}{\dimexpr\linewidth-2\fboxsep-2\fboxrule\relax}
\begin{enumerate}
\item Let $\widetilde \varphi=f(U,[0,\pi])$
\item If $\abs{\widetilde \varphi- \frac{\pi}{2}} > 2^{-\ell}\pi$,\\ \hspace*{0.5cm} run $f$ on $U$ in a range containing both $\widetilde \varphi$ and $2\pi-\widetilde \varphi$ (i.e.  $[\frac{\pi}{2},\frac{3\pi}{2}]$ or $[0,\frac{\pi}{2}]\cup [\frac{3\pi}{2},2\pi]$); return the result.
\item Else (Unable to determine whether $\varphi$ and $2\pi-\varphi$ are in $[\frac{\pi}{2},\frac{3\pi}{2}]$ or in $[0,\frac{\pi}{2}]\cup [\frac{3\pi}{2},2\pi]$)
	\begin{enumerate}
	\item Let $U'=e^{i\pi/4}U$
	\item Let $\widetilde \xi=f(U',[0,\pi])$
	\item If $\abs{\widetilde \xi- \frac{\pi}{2}} > 2^{-\ell}\pi$,\\ \hspace*{0.5cm} run $f$ on $U'$ in a range containing both $\widetilde \xi$ and $2\pi-\widetilde \xi$, (i.e.  $[\frac{\pi}{2},\frac{3\pi}{2}]$ or $[0,\frac{\pi}{2}]\cup [\frac{3\pi}{2},2\pi]$); return the result minus $\frac{\pi}{4}$. 
	\item Else, the algorithm fails.
	\end{enumerate}
\end{enumerate}
\end{varwidth}
}
{\\}

If every time $f$ is used, it finds the correct phase (or its mirror value) with accuracy $2^{-\ell}\pi$, then the algorithm successfully finds  $\varphi$ with accuracy $2^{-\ell}\pi$. $f$ is called at most three times, therefore by lemma \ref{lem:TernaryKitaev}  the algorithm above finds $\varphi$  with accuracy $2^{-\ell}\pi$ and with confidence greater than $1-3\ell e^{-m/160}$.

%
%

Once $\varphi$ is approximated for $U_{N,y}$, the corresponding energy under $H_{N,y}$ can be calculated by the relation $E=2\cos(\varphi)$ and the proof of theorem \ref{thm:ShorCounterexample} follows.

\renewcommand{\qedsymbol}{} 
\vspace{-\baselineskip}
\end{proof}


\section{Fast forwarding Hamiltonians and precision measurements} 
\label{sec:FF}
We would now like to identify the source of the ability to violate the cTEUP (postulate \ref{conj:cTEUP}). 
A careful examination of the proof of Theorem  \ref{thm:ShorCounterexample}
shows that the super-efficient 
energy measurement (SEEM) of the Hamiltonians $H_{N,y}$ is tightly connected to 
the ability to raise a number $y$ to an 
exponential power (modulo $N$, the number to be factorized), 
in polynomial time, i.e., to apply $U_{N,y}^{2^n}$ in polynomial time.
The straight forward way to do this would be to apply $U_{N,y}$ $2^n$ times sequentially. 
Instead, this is done in polynomial time using the unitary version of a classical procedure called 
\emph{modular exponentiation}. 
Equivalently, simulating the evolution of the quantum system according to the 
Hamiltonian $H'_y$ s.t. $U_{N,y}=e^{iH'_y}$,  for time $t=2^n$, 
can be done in time polynomial 
in $n$. We view this as \emph{fast forwarding} of the evolution 
governed by the Hamiltonian $H'_y$, as in Definition \ref{def:FF}.   
It turns out (as follows from the results of this section) 
that one can also achieve fast forwarding of the Hamiltonian 
$H_{N,y}$. We show in this section our main result, stating that
this is a general fact: 
fast forwarding (FF) the evolution governed by a given Hamiltonian, 
and exponentially accurate energy measurements with respect to this 
Hamiltonian, or what we call super efficient energy measurements (SEEM),  
are equivalent. Before stating the result more explicitly, and explaining 
the proof, we first provide an exact 
definition of the notion of SEEM.

\begin{definition} [SEEM] \label{def:SEEM} A normalized  Hamiltonian $H$ ($\norm{H}=1$) acting on $n$ qubits is $(\eta, \delta E, \beta)-$SEEM (super-efficient energy measurable) if there exist two unitaries $U_{SEEM},\widetilde{U}_{SEEM}$, 
acting on the $n$ qubits and on additional output/work qubits s.t.

\begin{enumerate}
\item $U_{SEEM}$ is a pre-measurement with $\eta$-accuracy $\delta E$ which doesn't perturb eigenstates of $H$, namely, 
\begin{equation}\label{eq:prem}
U_{SEEM}\ket{\psi_E,0,0}=\ket{\psi_E}\sum_{E
'}a_{E'}\ket{E',g(E')}
\end{equation} 
where $\ket{\psi_E}$ is an eigenvector of $H$ with eigenvalue $E$, 
$E'$ is the measurement outcome and $g(E')$ is some garbage left in the work register, and Equation \ref{eq:eta} is satisfied.   

\item The complexity of implementing $\widetilde{U}_{SEEM}$ is polynomial in $n$ and
\begin{equation}
\norm{U_{SEEM}-\widetilde{U}_{SEEM}}\le \beta
\end{equation}
\end{enumerate}
\end{definition}

We note that in fact, to be completely rigorous, the above definition, 
as well as that of FF, 
should consider a family of Hamiltonians $\{H_n\}_{n=1}^\infty$ rather than 
a given Hamiltonian. This will be clear from the context and is thus 
left implicit here and in the remainder of the paper.  

\begin{theorem}{\bf [Main]}\label{thm:main}
For $n$ the number of qubits, the following two sets of Hamiltonians are equivalent:
\begin{enumerate}
\item $\mathsf{FF}_{exp}$:  A normalized Hamiltonian $H$ acting on $n$ 
qubits is in $\mathsf{FF}_{exp}$ if there exists an exponentially growing
 function 
$T=O(2^{\Omega(n)})$ s.t. $H$ is $(T,\alpha)$-FF for  any $\alpha=1/poly(n)$.
\item $\mathsf{SEEM}_{exp}$: A normalized Hamiltonian $H$ acting on 
$n$ qubits is in $\mathsf{SEEM}_{exp}$ if there exists a function 
$\delta E=2^{-\Omega(n)}$ s.t. $H$ is $(\eta, \delta E, \beta)$-SEEM for any $\eta,\beta=1/poly(n)$.
\end{enumerate}
\end{theorem}

To prove the theorem, we will use two tools: the first tool 
(Lemma \ref{lem:ConfAmp}), 
gives efficient exponential confidence amplification  
of a low-demolition energy measurement,  
without increasing the demolition parameter $\beta$ too much
 (Confidence amplification). 
The second tool, Lemma \ref{lem:FFbyConcat}, allows increasing the $T$ parameter of fast forwarding at the cost of degrading $\alpha$ (FF by concatenation). 
That fast forwarding implies super efficient energy measurements is done 
by applying the FF by concatenation lemma, 
then applying the phase estimation circuit using this FF with improved 
parameters,   
and lastly improving the parameters by confidence amplification.
This proof could have been done using similar tools as in the proof of 
Theorem \ref{thm:ShorCounterexample}, but 
for variety we use here the more commonly used phase estimation 
circuit, which makes use of the quantum Fourier transform.   
To prove the other direction, namely that SEEM implies fast forwarding, 
we estimate the energy using the SEEM unitarily, apply the correct phase 
based on the result, and run the energy estimation backwards to erase garbage. 
Once again, the confidence amplification lemma (Lemma 2) is required in 
order to gain back the parameters which were degraded.  

We remark that Theorem \ref{thm:main} 
is stated as an equivalences between SEEM and 
FF for behaviors of $T$ and $\delta E$
which grow/decrease exponentially with $n$. 
In fact one can prove equivalences for other functions as well, 
however, there seems to be some inherent (constant) 
loss in parameters 
when moving between the notion of FF to within time $T$, and SEEM 
to within accuracy $1/T$;  
which is why the equivalence (Theorem \ref{thm:main}) is 
stated in terms of exponential functions rather than in terms of 
exact parameters. 

We start by stating and proving the two lemmas. 

\begin{lemma} [Confidence amplification] \label{lem:ConfAmp}
Let $\eta>\frac{1}{2}$, and let $H$ be a Hamiltonian on $n$ qubits, $\|H\|\le 1$, which is 
$(\eta, \delta E, \beta)-$SEEM. Then for any integer $m\ge 1$, $H$ is also $(1-e^{-\frac{m}{2}\prnt{1-\frac{1}{2\eta}}^2} , \delta E, m\beta)-$SEEM. 
\end{lemma} 
\begin{proof}
Consider $m$ applications of the non perturbing energy premeasurement unitary circuit $U_{SEEM}$ 
with $\eta$-accuracy $\delta E$. 
The probability that the majority of these outputs are within $\delta E$ of the correct 
value can be bounded by the Chernoff bound

\begin{equation}
\Pr(majority ~of~ measurements~ outside~ the~ window ~ \delta E ) \le e^{-\frac{m}{2}\prnt{1-\frac{1}{2\eta}}^2} 
\end{equation}

Hence a median of the measurements is at distance $\le\delta E$ from the correct energy value 
with confidence $1-e^{-\frac{m}{2}\prnt{1-\frac{1}{2\eta}}^2} $. 
We define the new premeasurement circuit $V_{SEEM}$ to first apply 
$U_{SEEM}$ $m$ different times, each time using a new ancilla register. Each such circuit writes 
$E'$ on its ancilla register.  $V_{SEEM}$ then unitarily computes 
the median of these $m$ outputs on an extra register. We know that had one of these outputs been measured, 
the probability that it is within $\delta E$ from the correct value $E$ is at least $\eta$. Since the measurements of those values 
mutually commute, are independent, and commute with the measurement of the 
median, we see that the median is 
within $\delta E$ from $E$ with probability at least   
$1-e^{-\frac{m}{2}\prnt{1-\frac{1}{2\eta}}^2} $. 

$\widetilde{V}_{SEEM}$ is defined by replacing 
$U_{SEEM}$ $m$ by 
$\widetilde{U}_{SEEM}$ in the above procedure. Since this is done 
$m$ times we have 
\begin{equation} 
\|\widetilde{V}_{SEEM}-V_{SEEM}\|\le m\|\widetilde U_{SEEM}-U_{SEEM}\|\le m\beta. 
\end{equation}
\end{proof}

The second tool allows increasing the $T$ parameter of fast forwarding at the cost of degrading $\alpha$. 
\begin{lemma} [FF by concatenation] \label{lem:FFbyConcat} For any 
integer $\kappa>0$, if a Hamiltonian is $(T,\alpha)$-FF, it is also 
$(T \kappa,\alpha \kappa)$-FF. 
\end{lemma}
\begin{proof}
The proof is by concatenation of $\kappa$ instances of the 
fast-forwarding circuit; the bound of $\alpha\kappa$ is derived by a 
standard telescopic argument. 
\end{proof}

To prove Theorem \ref{thm:main} we start by proving that fast forwarding implies super efficient energy 
measurements. After this we prove the other direction.

\begin{claim}  \label{cl:FF2SEEM}
For $T=O(2^{poly(n)})$, if a normalized Hamiltonian  on $n$ qubits is $(T,\alpha)$-FF, it is additionally
$(1-e^{-n/18},\frac{1}{T}, 16n\alpha\log(32T))$-SEEM.
\end{claim}

\begin{proof} 
We start by using the concatenation lemma (Lemma \ref{lem:FFbyConcat}) to claim the Hamiltonian is $(16T,16\alpha)$-FF. Next we show that 
$(16T,16\alpha)$-FF and $T=O(2^{poly(n)})$ $\Rightarrow$ $(\frac{3}{4},\frac{1}{T}, 16\alpha \log(32T))$-SEEM. 
The result then follows from the amplification lemma, Lemma \ref{lem:ConfAmp} with $m=n$.

We use the assumption that fast forwarding of $H$ is possible, 
to efficiently apply phase estimation with respect to the unitary 
$V=exp\prnt{i \prnt{H+\mathbbm{1}}}$. 
$V$ and $H$ of course share eigenvectors, and an eigenvalue $E$ of $H$ corresponds
to an eigenvalue $e^{i\varphi}$ for $V$ for $\varphi=E+1$ 
(recall that $\|H\| \le 1$ so $0\le \varphi=E+1\le 2\le 2\pi$).

Fix $\ell= \lfloor\log(32T)\rfloor$ to be the number of bits of $\varphi$ estimated in the 
phase estimation procedure. The procedure requires conditional applications of $\prnttt{V^{2^k}}_{k=0}^{\ell -1}$; 
This is done by implementing $\ell$ different instances of fast forwarding of $H$, $e^{iHt}$, 
with $t={2^{0},2^{1}\dots{2^{\ell-1}}}\le 16T$.

The following lemma is useful for evaluating the errors of phase estimation.
\begin{lemma} [phase estimation confidence (adapted from 5.2.1 in \cite{NC00})] \label{lem:PEConfidence}
Let $U$ be a unitary and $e^{i\varphi}$ an eigenvalue of $U$ and let 
an eigenvector with this eigenvalue be given as input to the phase 
estimation procedure. Let $m$ be the measurement outcome of an $\ell$-qubits 
phase estimation circuit (see figure \ref{fig:ShorPE}). For any $b+1< \ell$, 
\begin{equation}
\Pr\prnt{\abs{{\varphi}-\frac{2\pi m}{2^\ell}}>\frac{2\pi}{2^b}}\le \frac{1}{2(2^{\ell-b}-2)} \le \frac{1}{2^{\ell-b}}.
\end{equation}

\end{lemma}

\begin{figure} [H] 
\[
\Qcircuit @C=1em @R=.7em {
&&&&&&\sum_{k=0}^{2^{\ell}-1}\ket{k}&&\\
&&&\lstick{\ket{0}}     &\gate{\mathrm{Had.}}  	& \qw& \qw & \qw \cwx[3]  & \qw & \multigate{2}{QFT^{\dagger}} & \qw & \meter &\\
 \ell {\Biggl\{} &&{\hspace{.05cm}\vdots} 	& &\gate{\mathrm{Had.}}& \qw&\qw  &  \qw & \qw & \ghost {QFT^{\dagger}} &\qw & \qw {\hspace{0.25cm}\vdots}   \\
&&&\lstick{\ket{0}}		&\gate{\mathrm{Had.}}	& \qw& \qw & \qw & \qw & \ghost {QFT^{\dagger}}& \qw & \meter & \\
&&&\lstick{\ket{0\dots01}}	&\qw& \qw&\qw	& \gate{U_{N,y}^k=e^{-iH_yk}} &
}
\]
\caption{$\ell$-qubits phase estimation procedure.} \label{fig:ShorPE}
\end{figure}
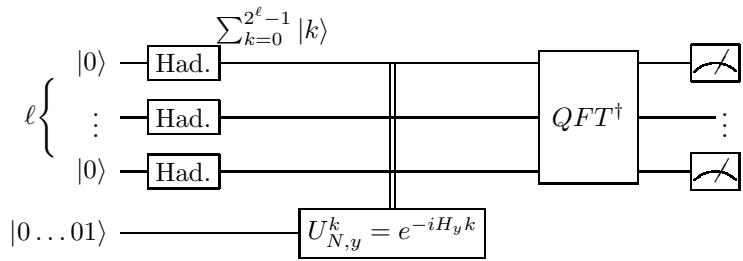

Using Lemma \ref{lem:PEConfidence}, we get that the $\ell$-bit phase estimation procedure 
estimates $\varphi$ to within 
 $\delta\varphi=\pi\cdot2^{-(\ell-3)}$ with confidence $3/4$. 
We get that the procedure provides an outcome 
which is within  $\delta E=\delta\varphi=4\pi\cdot 2^{-(\ell-1)}\le \frac{4\pi}{16T}< \frac{1}{T}$ from $E$ with confidence $3/4$. 

To apply the $\ell$ instances of conditional applications of powers of $V$;
$\{V^{2^k}\}_{k=0}^{\ell-1}$, we apply $\ell$ different $16\alpha$-approximations of 
$e^{iHt}\otimes \mathbbm{1}_{2^c}$ (using the fast forwarding) where each such application works 
on the state plus its own ancilla register initialized to $0$ 
(as in Definition \ref{def:FF}). 
We get that $\beta\le 16 \alpha \ell \le 16\alpha \log(32T))$. 
\end{proof}

\begin{corollary}
 $\mathsf{FF}_{exp} \subseteq \mathsf{SEEM}_{exp}$
\end{corollary}
\begin{proof}
A Hamiltonian $H\in \mathsf{FF}_{exp}$, can be FF for some $T=2^{poly(n)}$, with $\alpha=\frac{\beta}{16n\log(32T)}=O\prnt{\frac{1}{poly(n)}}$ for any $\beta=O\prnt{\frac{1}{poly(n)}}$. Hence, by 
Claim \ref{cl:FF2SEEM}, $H$ is $(1-e^{-n/18},1/T,\beta)$-SEEM, and therefore it is $(\eta,1/T,\beta)$-SEEM for any $\eta,\beta=O\prnt{\frac{1}{poly(n)}}$. We conclude that 
$H\in \mathsf{SEEM}_{exp}$.
\end{proof}

We now prove that SEEM implies FF with the desired parameters: 

\begin{claim} \label{cl:SEEM2FF}
Let $H$ be an $n$ qubit Hamiltonian with $\|H\|\le 1$ which is $(\eta,\delta E, \beta)$-SEEM for $\eta>1/2$. 
Let $T\delta E<\frac{\pi}{2}$, 
 then $H$ is also $(T,~ 2\eta\sin(\delta E T)+2(1-\eta+\beta) )$-FF. 
\end{claim}

\begin{proof}
The idea of the proof is to 
apply the unitary $\widetilde U_{SEEM}$ approximating the 
premeasurment of the energy, which exists since the Hamiltonian can be super-efficiently measured, 
by Definition \ref{def:SEEM}. 
Then, based on the output $E'$ of this premeasurement, written on the quantum register, 
multiply the state by the phase $e^{-iE't}$ 
(denote this by the gate $V$), and finally apply the inverse of the approximated premeasurement unitary. 
Let $\ket{\alpha}=\ket{\psi_E}\otimes\ket{0}$. 
First we consider the exact premeasurement with no demolition ($\beta=0$), $U_{SEEM}$; notice that it commutes with $H$:  
\begin{equation}
\begin{split}
\norm{\prnt{U_{SEEM}^\dagger VU_{SEEM}-e^{-iHt}\otimes \mathbbm{1}_{\mathcal{W}}}\ket{\alpha}}
=&
\norm{\prnt{U_{SEEM}^\dagger VU_{SEEM}-U_{SEEM}^\dagger \prnt{e^{-iHt}\otimes \mathbbm{1}_{\mathcal{W}}}U_{SEEM}}\ket{\alpha}}\\
=&
\norm{\prnt{VU_{SEEM}- \prnt{e^{-iHt}\otimes \mathbbm{1}_{\mathcal{W}}}U_{SEEM}}\ket{\alpha}},
\end{split}
\end{equation}
where the Hilbert space of the work/output register is denoted by $\mathcal{W}$. 
On a specific eigenvector $\psi_E$:
\begin{equation}
\begin{split}
&\norm{ VU_{SEEM}\ket{\psi_E,0,0}- \prnt{e^{-iHt}\otimes \mathbbm{1}_{\mathcal{W}}}U_{SEEM}\ket{\psi_E,0}}=\norm{\sum_{E'} a_{E'} (e^{-iE't}-e^{-iEt}) \ket{\psi_E,E',g(E')}}\\
&=\norm{\sum_{E':\abs{E'-E}\le\delta E} a_{E'} (e^{-iE't}-e^{-iEt}) \ket{\psi_E,E',g(E')}+\sum_{E':\abs{E'-E}>\delta E} a_{E'} (e^{-iE't}-e^{-iEt}) \ket{\psi_E,E',g(E')}}\\
&\le 2\eta\sin(\delta Et) + 2(1-\eta), 
\end{split}
\end{equation} 

where the last inequality is correct for $t\le \pi/2\delta E$.
Notice that the above holds for any state $\ket{\psi}=\sum_{c_E} c_E \ket{\psi_E}$, 
using the fact that both $U_{SEEM}$ and $V$ leave the left register in tact.  
The proof follows since we have $\norm{\widetilde{U}_{SEEM}^\dagger V\widetilde{U}_{SEEM}-U_{SEEM}^\dagger VU_{SEEM}}\le 2\beta$.
\end{proof}

\begin{corollary} \label{cor:EEM2FF}
Let $H$ be a normalized Hamiltonian on $n$ qubits, which is $(\eta,\delta E,\beta)$-SEEM for $\eta>1/2$ and $\beta<\pi/2$. Then $H$ is also $(\beta/\delta E,2n\beta+2^{-O(poly(n))})$-FF.
\end{corollary}

\begin{proof}

Using lemma \ref{lem:ConfAmp} with $m=n-1$ we reach an  $(1-e^{-\frac{(n-1)\eta}{2} \prnt{1-\frac{1}{2\eta}}^2},\delta E,(n-1)\beta)$-SEEM. 
Now choose $T=\beta/ \delta E$ 
and since $T\delta E=\beta<\pi/2$ we can apply claim \ref{cl:SEEM2FF}.
The FF error $\alpha$ according to Claim \ref{cl:SEEM2FF} is bounded by $2n\beta+2^{-poly(n)}$.

\end{proof}

\begin{corollary}
$\mathsf{SEEM}_{exp} \subseteq \mathsf{FF}_{exp} $
\end{corollary}
\begin{proof}
Let $H\in \mathsf{SEEM}_{exp}$, with some inverse exponential 
function $\delta E$. We choose $T$ to be any exponentially growing
function such that $T\delta E$ decays faster than any polynomial (say, 
$T=\frac{1}{\delta E^{0.99}}$).  
Let $\alpha=O\prnt{\frac{1}{poly(n)}}$ be a goal parameter for the fast forwarding.  By assumption, $H$ is $(\eta=2/3,\delta E, 
\beta=\alpha/3n=O\prnt{\frac{1}{poly(n)}})-SEEM$. 
From corollary \ref{cor:EEM2FF}, $H$ is also  $(\beta/\delta E,\alpha)$-FF. 
By our choice of $T$, it is thus $(T,\alpha)$-FF. 
Since this holds for any inverse polynomial $\alpha$, 
we have $H\in \mathsf{SEEM}_{exp}$. 
\end{proof}
This completes the proof of theorem \ref{thm:main}.

\section{Exponential Fast Forwarding for Physical Hamiltonians} \label{sec:FFPhysicalHamiltonians}

\subsection{Commuting local Hamiltonians} \label{sec:MBL}
A class of Hamiltonians that can be easily fast forwarded are commuting local 
Hamiltonians. A Hamiltonian $H$ is a commuting $k$-local Hamiltonian
if it is of the form
\begin{equation}
H=\sum_j H_j,
\end{equation}
where every term $H_j$ acts non-trivially on at most on $k$ qubits, 
and $[H_i,H_j]=0$ for all $i,j$. 

\begin{theorem}\label{thm:comm}
If $H$ is an $n$ qubit normalized commuting $k-$local Hamiltonian, 
with $k=O(\log(n))$, then it can be $(T,\alpha)$-fast forwarded 
with $T=2^{\Omega(n)}$ and arbitrary exponentially small  $\alpha$. 
\end{theorem} 
\begin{proof}
Since $H_j$ commute, we have 
\begin{equation} 
e^{-iHt}=\prod_j e^{-iH_jt}. 
\end{equation} 
It thus suffices to be able to implement $e^{-iH_jt}$ for $t$ exponentially 
large, with an appropriate exponentially small accuracy. 
 Given the description of $H_j$, let $U_j$ be the matrix which 
diagonalizes it, 
\begin{equation}
D_j=U_j H_j U^\dagger_j
\end{equation} 
All entries of $U_j$ can be efficiently calculated 
classically to within exponential accuracy, 
since the dimensions of the matrix are $2^k\times 2^k$ for $k=O(\log(n))$; 
its associated eigenvectors $\ket{\psi_m}$ and most importantly the 
corresponding 
eigenvalues $\lambda_m$ can also be calculated classically with exponential 
accuracy. This follows from classical results about matrix 
calculations, see e.g., \cite{PC99,ABBCS15}. 
Since the dimension of $U_j$ is polynomial, and all its entries 
are known, 
$U_j$ can be applied efficiently by a quantum computer, 
on the appropriate $k$ qubits, to achieve the transformation 
$\ket{\psi_m}\longmapsto \ket{m}$, for 
$H_j$.  
Then using a quantum computer one can apply the transformation 
\begin{equation}
\ket{m}\longmapsto e^{-i\lambda_m t}\ket{m}
\end{equation} 
to within exponentially good accuracy; and then the quantum computer 
can apply $U_j^\dagger$, which all together achieves $e^{-iH_jt}$ to within 
exponentially good accuracy.  
\end{proof} 

The toric code Hamiltonian constitutes \cite{Kitaev2003}
an important example for a commuting $4$-local Hamiltonian;  
We note that unlike what one might suspect, 
the time evolution of commuting local Hamiltonians generates 
very interesting behavior from the computational perspective;
even evolving for a unit time results in generating distributions which 
are hard to even approximately simulate classically (under commonly 
believed computational assumptions). See \cite{BJS10,BMS16}. 

We stress that the important part of Theorem \ref{thm:comm} is the 
fact that fast forwarding is possible for an {\it exponentially long time}; 
The fact that exponentially good accuracy $\alpha$ is achievable is 
less important. Even simulating a physical process for exponential time 
with a {\it small constant error} would be a remarkable achievement, 
since we could repeat the experiment polynomially many times and gain 
as much information about the results as we want. 

\subsection{Fast forwarding quadratic Hamiltonians} \label{sec:quadratic}

Using a similar idea to the above, one can derive exponential 
fast forwarding for 
a wide class of physically interesting Hamiltonians, called quadratic Hamiltonians. An important special case of this class is Anderson's model for 
electron localization \cite{Anderson58}.
 
We provide this here for the case of Fermions. 
We presume a quadratic Hamiltonian of bosons can also be fast forwarded, 
but haven't checked it. The states of $n$ indistinguishable 
Fermions distributed 
over $m=poly(n)$ modes are described by Fock space \cite{BookIZ06}
of possibly exponential dimension in $n$: ${m\choose{n}}$ in case of fermions, or ${n+m-1 \choose m-1}$ in case of bosons. 
A quadratic Hamiltonian is defined as follows:  
\begin{equation} \label{eq:GenericQuadraticH}
H=\sum_{i,j}^m A_{i,j} a^\dagger_i a_j +\frac{1}{2}\sum_{ij}B_{i,j}a_i a_j +\frac{1}{2}\sum_{i,j}B_{j,i}^* a_i^\dagger a_j^\dagger ~~~~~~ A=A^\dagger,B=B^\dagger
\end{equation}
The Hermiticity of $H$ follows from the fact that the 
matrices $A$ and $B$ are Hermitian.
$a_i, a_i^{\dagger}$ are the creation and annihilation operators, 
satisfying the anti commutation relations for fermions,
\begin{equation} \label{eq:FermionsACom}
\prnttt{a_i, a_j}=0 ~~~~ \prnttt{a_i^\dagger, a_j^\dagger}=0~~~~\prnttt{a_i, a_j^\dagger}=\delta_{i,j}. 
\end{equation} 
$a_i^\dagger a^\dagger_j=-a^\dagger_j a^\dagger_i$ hence equation \ref{eq:GenericQuadraticH} takes the form
\begin{equation}
H=\sum_{i,j} A_{i,j} a^\dagger_i a_j +\frac{1}{2}\sum_{ij}B_{i,j}a_i a_j -\frac{1}{2}\sum_{i,j}B_{i,j}^* a_i^\dagger a_j^\dagger.  
\end{equation}

We now assume that we can physically implement any quadratic Hamiltonian, 
s.t. the error in each coefficient is at most inverse polynomial.  
Then: 
\begin{theorem}\label{thm:quadFF}
Let $H$ be a quadratic Hamiltonian of $n$ Fermions
 with $poly(n)$ modes. $H$ can be $(T,\alpha)$-fast forwarded 
with $T=2^{\Omega(n)}$ and arbitrary inverse polynomial $\alpha$.
\end{theorem} 
\begin{proof}
The proof idea is to efficiently ``diagonalize'' the traceless part of the Hamiltonian by the Bogolyubov transformation \cite{Bblaizot86,Shchesnovich13} to the form $H=\sum_i \lambda_i b_i^\dagger b_i$. The operators $b_i,b_i^\dagger$ are called quasiparticle annihilation  and creation operators respectively, and they 
inherit the commutation/anti commutation relations of $a_i,a_i^\dagger$ as in Equations \ref{eq:FermionsACom}. 
Additionally, the number operator $b^\dagger_ib_i$ has integer eigenvalues. 
Fast forwarding is enabled by efficiently evolving the system under $H'=\sum_i (\lambda_i t \mod 2\pi )  b_i^\dagger b_i$ for one time unit and adding a global phase $e^{-i\mathrm{tr}(A)t/2}$.

We now describe the details, using standard claims in Physics (whose proofs 
can be found in the Appendix for completeness): 

\begin{claim} \label{clm:QuadraticMatrixForm}
Let $\textbf{a}$ be a
column vector whose $j\th$ coordinate is $a_j$ and let $\textbf{a}^\dagger$ be  a column vector whose $j\th$ coordinate 
is $a_j^\dagger$. The Hamiltonian $H$ can be written as 
\begin{equation}
H=\frac{1}{2}\left(\begin{array}{cc}
\overline{\mathbf{a}^{\dagger}} & \overline{\mathbf{a}}\end{array}\right)\left(\begin{array}{cc}
A & B^{*}\\
B & -A^{*}
\end{array}\right)\left(\begin{array}{c}
{\mathbf{a}}\\
{\mathbf{a}}^{\dagger}
\end{array}\right)+\frac{1}{2}\tr(A)
\end{equation} 
Here, the overline  indicates a matrix transposition, i.e., $\overline{\mathbf{a}^{\dagger}}, \overline{\mathbf{a}} $ are the row vectors corresponding to $\textbf{a}^\dagger, \textbf{a}$ respectively.
\end{claim}

\begin{claim}
The traceless part of the Hamiltonian can be diagonalized: 
\begin{equation}
H=  \frac{1}{2}\left(\begin{array}{cc}
\overline{\mathbf{a}^{\dagger}} & \overline{\mathbf{a}}\end{array}\right)U D U^\dagger \left(\begin{array}{c}
{\mathbf{a}}\\
{\mathbf{a}}^{\dagger}
\end{array}\right)+\frac{1}{2}\tr(A)
\end{equation} 
where $D$ is a real diagonal matrix s.t. $D_{j,j}=-D_{j+n,j+n}$ and $U$ is unitary. Furthermore, there exist matrices $V_1,V_2$ s.t. $U$ is a block matrix in the form $U=\left(\begin{array}{cc}
V_1 & V_2^*\\
V_2 & V_1^*
\end{array}\right) $
\end{claim}

\begin{claim}\label{cl:integers}
By defining
\begin{equation} 
\left(\begin{array}{c}
{\mathbf{b}}\\
{\mathbf{b}}^{\dagger}
\end{array}\right)=
U^\dagger \left(\begin{array}{c}
{\mathbf{a}}\\
{\mathbf{a}}^{\dagger}
\end{array}\right).
\end{equation}
The Hamiltonian takes the form
\begin{equation}  \label{eq:DecopuledModesFermions}
H= 2\sum_{i=1}^m b_i^\dagger b_i D_{i,i} + \frac{1}{2} \tr (A)
\end{equation}
The new operators obey the anti-commutation relations of fermions. In addition, the eigenvalues of $b_i^\dagger b_i$ either 0 or 1.  
\end{claim}

We can now use the above claims to achieve fast forwarding. 
In Equation \ref{eq:DecopuledModesFermions}, the modes $b_i$ are independent, and in particular, 
the Hamiltonian is a sum of $m$ commuting terms. 
Therefore, an evolution under $H$ for time $t$ can be implemented by 

\begin{equation}
e^{iHt}= e^{-it\tr(A)/2} \prod_j e^{2iD_{jj}tb_j^{\dagger} b_j}
\end{equation} 
Since the eigenvalues of each term are $D_{j,j}t$ times an integer 
(using Claim \ref{cl:integers})  
we have that if we replace the traceless part of the Hamiltonian by a matrix with the 
same eigenvectors but with eigenvalues $(D_{j,j}t \mod 2\pi)$, it will have the same 
evolution on any state, in other words if we define 
\begin{equation}\label{eq:quadHFermions}
 H'=\frac{1}{2}\left(\begin{array}{cc}
\overline{\mathbf{a}^{\dagger}} & \overline{\mathbf{a}}\end{array}\right)U D' U^\dagger \left(\begin{array}{c}
{\mathbf{a}}\\
{\mathbf{a}}^{\dagger}
\end{array}\right) 
\end{equation}
 \begin{equation}
D'_{i,j}= \delta_{i,j} D_{i,i} t \mod 2\pi
\end{equation}
we get 
\begin{equation}
 e^{-i H'}\cdot e^{-it\tr(A)/2}= e^{-iHt}
\end{equation}
Hence it is sufficient to simulate the evolution under $H'$ to time $t'=1$ and to add the global phase $e^{-it \tr(A)/2}$. 
To apply $H'$ for $t'=1$, observe that Equation \ref{eq:quadHFermions} means that 
$H'$ is a quadratic Hamiltonian in $\prnttt{a_i}$ and $\{a_i^\dagger\}$,
whose coefficients can be calculated by a classical computer
in time polynomial in $m$ to 
within exponential accuracy \cite{PC99,ABBCS15}. 
Assuming that we can implement any quadratic Hamiltonian of polynomial 
number of coefficients exactly, we need only apply $H'$ for one time unit
to fast forward $H$ for exponential duration $t$, with 
arbitrary exponentially small $\alpha$.
However, the assumption that a quadratic Hamiltonian with 
general coefficients can be implemented exactly is not realistic; 
Assuming inverse polynomial error in each of 
the coefficients in the quadratic Hamiltonian results in an overall 
inverse polynomial error and thus would still lead to a fast 
forwarding procedure for exponential duration of time, but with   
inverse polynomial error $\alpha$. 
\end{proof}

\section{Impossibility of a Generic Fast Forwarding Procedure for 
realizable Hamiltonians\label{sec:2SparseFFimpBQPinPSPACE}}

Perhaps any physically 
realistic Hamiltonian (one which can 
be efficiently simulated by a quantum circuit) can be fast-forwarded? 
We show that this is highly unlikely: 
even if a subset of the Hamiltonians - namely, the 2-sparse row-computable 
Hamiltonians - can be exponentially fast forwarded, the complexity class PSPACE 
equals BQP (which is highly unlikely). 

$2$-sparse row-computable Hamiltonians (as mentioned in Section \ref{sec:ShorAsCounterexample}) are Hamiltonians with at  
most two non-zero 
entries per row, and 
such that given the row number, it is possible (either quantumly or 
classically) to efficiently compute the column indices and values of 
the non-zero entries of this row. It is known that such Hamiltonians are 
efficiently simulable by a quantum circuit \cite{AT03,BCCKS14}. 

\begin{theorem} \label{thm:No2sparseFF}
A generic procedure for $(T=2^{(n^{1/c})},\alpha=n^{-4/c})$-fast 
forwarding a 2-sparse row computable Hamiltonians, with $c>1$, does not 
exist (unless
$BQP=PSPACE$).
\end{theorem}

The structure of the proof is as follows: we first assume a
fast forwarding procedure with better parameters, 
namely we start with a $(T=5^n, \alpha=n^{-4})-$FF procedure, 
and utilize it to design a polynomial time
quantum algorithm which solves the PSPACE-complete problem 
\texttt{OTHER END OF THIS LINE (OEOTL)} \cite{Papad94}. We will later show that 
the parameters in the statement of the theorem suffice to imply 
the assumed parameters, essentially by 
a simple padding argument.  
 
Let $G=(V,E)$  be a directed graph where every vertex is represented by an $n$ bits vectors ($2^n$ vertices). The edges are represented by two polynomial size circuits $S$ and $P$. There is an edge from $u$ to $v$ only if $S(u)=v$ and $P(v)=u$; hence $G$ contains paths, cycles, or isolated vertices. 

\begin{definition}  [\texttt{OEOTL}] \label{def:OEOTL}
Given that the vertex $0^n$ has no incoming edge but has an outgoing edge, find the other end of the line that starts with $0^n$.
\end{definition}

Let $A_G$ be the adjacency matrix of the undirected graph induced by $G$. 
We use $H=\frac{1}{2}A_G$ 
as a $2$-sparse row-computable 
 Hamiltonian $H$ ($\norm{H}=1$). 
Without loss of generality, we denote the vertices along the path by 0...$L-1$ with $L\le 2^n$. Evolving the the state $\ket{0}$ under  $H$ keeps the system in the subspace of all vertices in the path. The eigenvalues and eigenvectors of the Hamiltonian in this subspace are:
\begin{gather}\label{eq:evpspace} 
E_k=\cos\prnt{\frac{\pi k}{L+1}} ~~~~ k=1...L\\\label{eq:evpspace2} 
\ket{\psi_k}=\sqrt{\frac{2}{L+1}}\sum_{j=0}^{L-1} \sin \prnt{\frac{\pi k(j+1) }{L+1}}\ket{j}, 
\end{gather}
which is easy to check: 
\begin{equation}
H\ket{\psi_k}=\frac{1}{2}\sqrt{\frac{2}{L+1}}\sum_{j=0}^{L-1} \prntt{\sin \prnt{\frac{\pi kj }{L+1}}+\sin \prnt{\frac{\pi k(j+2) }{L+1}}}\ket{j}=\sqrt{\frac{2}{L+1}}\sum_{j=0}^{L-1} \prntt{\sin \prnt{\frac{\pi k(j+1) }{L+1}}\cos \prnt{\frac{\pi k }{L+1}}}\ket{j}. 
\end{equation}

We assume that a generic $(T= 5^n, \alpha= n^{-4})$-FF 
is possible for any 2-sparse row computable Hamiltonian.
By Claim \ref{cl:FF2SEEM} this means that 
SEEM for this Hamiltonian is possible with parameters $\eta= 1-e^{-n/18}$, $\delta E=5^{-n}$, $\beta=40n^{-2}$.  By Equation \ref{eq:evpspace} 
at most one eigenvalue fits in a window of size $2\delta E$. This is because the minimal gap is between $k=1,2$ (or $k=L-1,L$).  
$\cos (\frac{\pi}{L+1}) - \cos (\frac{2\pi}{L+1}) = 2 \sin(\frac{\pi}{2(L+1)})\sin (\frac{3\pi}{2(L+1)})\ge \frac{3\pi^2}{8(L+1)^2}\ge \frac{2}{(L+1)^2}\ge 2\delta E$, the leftmost inequality is due to $\sin x > x/2$ for $0<x<\pi/4$.
Hence, estimating the energy of any state to within such $\delta E$ should intuitively approximately 
project the state onto an eigenvector. 

The polynomial time algorithm is as follows:

{~}

\noindent\fbox{%
\begin{varwidth}{\dimexpr\linewidth-2\fboxsep-2\fboxrule\relax}
\begin{enumerate}
\item Let $v_0=0$, $H_0=H$
\item For $i$=1 to $100n$ 
	\begin{enumerate}
\item Check that $10$ steps forward from $v_{i-1}$  the end of the line is 
not reached; if it is, output it and exit. 
	\item Perform a $(1-e^{-n/18},5^{-n},40n^{-2})$-SEEM 
 on the state $\ket{v_{i-1}}$ under $H_{i-1}$.\label{itm:AlgEMeas} 
	\item  \label{itm:AlgVMeas} Measure in the vertices basis, denote the result by $v_i$.
Let $H_i$ be the original Hamiltonian $H$ with the edge $(v_i-1,v_i)$ removed. 
	\end{enumerate}
\item Move $10$ steps forward from the vertex reached. If the end of the line is found return it and exit. Otherwise, the algorithm fails.\label{itm:10steps}
\end{enumerate}
\end{varwidth}
}

{~}

\noindent \textbf{Proof of correctness:}\\
The idea of the proof is to make progress on the path, as follows: 
Starting from the first node $v_0$, we use the SEEM (stage \ref{itm:AlgEMeas})
to measure the energy with respect to $H_0$ (assume $\beta=0$ for now).  
Due to the high accuracy of the measurement, 
the resulting state, conditioned on the measurement outcome, 
is close to an eigenstate. All eigenstates are symmetric around 
the middle of the path, hence the measurement in the vertices basis 
in stage \ref{itm:AlgVMeas}, yields with good probability a vertex $v_1$ that
 is closer to the end of the path than to $v_0$ namely, 
the remaining path length is likely to be halved. 
We call this event a successful iteration, and show its probability is at least $1/10$ if the path length is more than 10. An unsuccessful 
iteration does not increase the length of the path, it just 
doesn't succeed in shrinking by half. 
The vertex $v_1$ is now the 
next starting point, and $H_1$ is fixed to prevent going backwards by 
correcting  $H$ to not include the edge connecting  $v_1$ to the previous 
vertex on the line. After $n$ 
successful iterations, 
the vertex reached should be the end of the line.
By Chernoff, the probability of at least $n$ 
successful iterations with $1/10$ success probability, out of $100n$ 
iterations is $\ge 1-e^{-81n/20}$,  which is exponentially 
close to one. If at some stage the length of the path is smaller than 10, the success probability may be smaller than 1/10, but the end of the line is found in stage \ref{itm:10steps}. 
An analysis of $\beta>0$ concludes the proof.

Let $U_{meas.}$ be the pre-measurement with $\beta=0$, and confidence 
$\eta= 1-e^{-n/18}$ with respect to $H_i$.
Suppose the vertex found in the previous round is $v=v_i$, 
and consider applying 
$U_{meas.}$ to $\ket{v,0,0}$ where the additional two registers are the output and work registers.  
Denote the result of measuring the two additional registers by $\varepsilon_j, g$.   Let $f$ be the function s.t. $\psi_{f(j)}$ is the eigenstate of $H_i$ with energy closest to $\varepsilon_j$ 
(the lower energy eigenstate if there is a tie).  $f$ is well defined 
since all eigenvalues of $H_i$ within the 
relevant subspace 
have multiplicity 1. We omit adding an $i$ index to $f$, and to the eigenstates/eigenvalues of $H_i$ since both $H$ and $H_i$ in the relevant subspaces are Hamiltonians of paths, only the length of the path and the starting vertex change.
\begin{claim} \label{cl:ExpectationOfDominantEigvec}
Let $a_j$ be the amplitude of $\psi_{f(j)}$ after measuring $\varepsilon_j,g$. The expectation of $\abs{a_j}^2 $ over $j,g$  satisfies: 
$
\mathbb{E}_{j,g}\prnt{\abs{a_j}^2}=\sum_{g,j} \abs{a_j}^2 \Pr \prnt{\varepsilon_j,g} \ge \eta
$.

\end{claim}

\begin{proof}

Let $f^{-1}$ be the preimage of $f$, 
\begin{equation}
\begin{split}
\mathbb{E}_{j,g} (\abs{a_{j}}^2)&=
\sum_{g,j} \abs{a_{j}}^2 \Pr \prnt{\varepsilon_j,g}=
\sum_{g,j} \Pr(\psi_{f(j)} | \varepsilon_j,g) \Pr \prnt{\varepsilon_j,g}=
\sum_{g,j} \Pr( \varepsilon_j,g,\psi_{f(j)})\\
&= \sum_k \Pr (\psi_k) \sum_{g,j: j\in f^{-1}(k)} \Pr(\varepsilon_j,g | \psi_k) \ge \eta
\end{split}
\end{equation} 
The last inequality, is due to the $\eta$-confidence of the measurement, and that all measurement outcomes in the window $\delta E$ around $E_k$ are in $f^{-1}(k)$.
\end{proof}

\begin{claim} \label{cl:half}
Let $L_i=L-v_i\ge 10$ and 
$\ell_i =\left\lceil\frac{L_i}{2} \right\rceil$.  The probability for a successful iteration, i.e., $v_{i+1}\ge v_i + \ell_i$ is at least $1/10$ for the value of $\eta$.
%
\end{claim}
\begin{proof}
After measuring $\varepsilon_j,g$, the state of the system is  
$a_{j} \ket{\psi_{f(j)}} + \sqrt{1-\abs{ a_{j}}^2} \ket{\psi_{f(j)}^\perp}$. We define the $\Pi_i$ to be a projection on the vertices $v\ge v_i+\ell_i$.
The symmetry of the eigenstates around the middle of the path implies that $2\norm{\Pi_i\ket{\psi_{f(j)}}}^2+\frac{2}{L_i+1}\ge 1$, therefore $\frac{1}{2} - \frac{1}{L_i+1} \le \norm{\Pi_i \ket {\psi_{f(j)}}}^2\le \frac{1}{2}$. 
\begin{equation}
\begin{split}
\Pr \prnt{v_{i+1}\ge v_i + \ell_i| \varepsilon_j,g} &= \norm{\Pi_i \prnt{{ a_j}\ket{\psi_{f(j)}}+\sqrt{1-\abs{a_j}^2}\ket{\psi_{f(j)}^\perp}}}^2 
\\
&\ge 
\abs{a_j}^2 \norm{\Pi_i  \ket{\psi_{f(j)}}}^2 + \prnt{1-\abs{a_j}^2} \norm{\Pi_i \ket{\psi_{f(j)}^\perp}}^2 - 2 \abs{a_j} \sqrt{1-\abs{a_j}^2} \norm{\Pi_i \ket{\psi_{f(j)}}} \norm{\Pi_i \ket{\psi_{f(j)}^\perp}}
\\
&\ge \abs{a_j}^2 \prnt{\frac{1}{2}-\frac{1}{L_i+1}} - \abs{a_j} \sqrt{2-2\abs{a_j}^2}
\end{split}
\end{equation}
Using the inequality $x\sqrt{2-2x^2} \le 99(1-x^2) +0.01$ for $0\le x \le 1$, we bound the probability by $\abs{{a}_j}^2\prnt{99.5-\frac{1}{L_i+1}}-99.01$. 
Finally, we average over all $g,j$:
\begin{equation}
\begin{split}
\Pr \prnt{v_{i+1}\ge v_i+\ell_i}
&\ge
\sum_{g,\varepsilon_j} \Pr \prnt{v_{i+1}\ge v_i+\ell_i | \varepsilon_j,g} \Pr (\varepsilon_j,g) \ge  \mathbb{E}_{j,g}\prnt{ \abs{ a_j}^2} \prnt{99.5-\frac{1}{L_i+1}}-99.01\\
&\ge \eta\prnt{99.5-\frac{1}{L_i+1}}-99.01
\end{split}
\end{equation}
With  $\eta=1-e^{-n/18}$, the probability for a successful iteration is at least $1/10$ for $L_i>10$. 
\end{proof}

\begin{claim} 
When the algorithm applies the SEEM with $\beta=0$, it succeeds
with probability at 
least $1-e^{-81n/20}$. 
\end{claim}
\begin{proof}
Consider the  $100n$ 
iterations in the protocol, when if the protocol had ended before completing 
all iterations, the iteration is simply idle. 
An iteration $i$ is declared ``successful'' if either it is idle, or if not, 
the length of the path had been halved during 
step 2(c) of this iteration. 
By this definition and by Claim \ref{cl:half}, 
the probability of the $i$\th iteration 
to be successful  
is $>1/10$, even when we condition on what happened in previous iterations. 
Let $X$ be the number of successful 
iterations out of the $100n$ iterations.
We want to bound the probability that $X>n$ from below. 
We note that this probability is bounded from below by
the corresponding probability for the number $Y$ of successful iterations 
when 
we have $100n$ i.i.d Bernoulli variables, each with 
probability exactly $1/10$ for success.  
For i.i.d variables we can use the Chernoff bound,    
\begin{gather}
\Pr(Y\le (1-\delta)\mu ) \le e^{-\delta^2\mu /2} .
\end{gather}
where $\mu$ is the expectation of $Y$, and we have here $\mu = 10n$. Setting 
$(1-\delta)\mu = n$, $\delta = 9/10$, we have 
$\Pr(Y>n)\ge 1-e^{-81n/20}$. 
This means that $Pr(X>n)$, the probability for at least $n$ successful 
iterations, is at least 
$1-e^{-81n/20}$. 
After $n$ successful iterations the path length must have reached 
below $10$ since $\frac{L}{2^n}\le  1$, and the algorithm succeeds in finding 
the end of the line. 
\end{proof}

We analyze what is the success probability of the algorithm with $\beta>0$. 
To this end, consider first a unitary version of the above algorithm, still with $\beta=0$, where the only measurement is at the end. 
In stage \ref{itm:AlgVMeas}, we copy (using cNOTs) the result $v_i$ to a separate register in every iteration instead of measuring. In stage \ref{itm:AlgEMeas} we apply the super efficient energy pre-measurement, where the Hamiltonian is conditioned on the copy of $v_i$. The outcome and the garbage are written on a separate register in every iteration. At the end of the algorithm, an indicator qubit is set to 1 if the algorithm found the end of the line and 0 otherwise. In this version the algorithm is unitary, and the only difference between $\beta=0$ and $\beta>0$ cases are the $100n$ instances of SEEMs. Thus, at the end of the algorithm, just before measuring, the state with $\beta>0$ SEEMs (denoted $\widetilde \xi$) deviates at most by $100\beta n$ from the state in which $\beta=0$ SEEMs were used (denoted $\xi$). Let $\Pi$ be the projection on a successful outcome of the algorithm, i.e., the indicator qubit is 1. We bound the algorithm success probability for $\beta>0$:
\begin{equation}
\norm{\Pi \ket{\widetilde \xi}}^2=\norm{\Pi \ket{\xi}+ \Pi \prnt{\ket{\widetilde \xi}-\ket{\widetilde \xi}}}^2 \ge \prnt{\norm{\Pi \ket{\xi}} - 100\beta n}^2 \ge  \norm{\Pi \ket{\xi}}^2 - 200\beta n
\end{equation}
Hence the success probability is reduced by $200\beta n$. Since $\beta=40n^{-2}$ the success probability is polynomially close to 1.

The main contribution to the time complexity of the algorithm is from the $100n$ rounds of step 3. The SEEM in each round calls $O(\log n)$ times to the fast forwarding procedure. Hence the time complexity is polynomial in $n$.

We conclude the proof of Theorem \ref{thm:No2sparseFF} by relaxing the demand for the generic FF procedure. 
Consider a generic fast forwarding procedure for $n$ qubit Hamiltonians with parameters $T=2^{(n^{1/c})}$, and $\alpha = n^{-4/c}$ with $c>1$. This procedure is weaker than a generic fast forwarding procedure  with parameters $T=5^n$ and $\alpha=n^{-4}$ used in the proof, however one can use the weaker procedure to $(5^n,n^{-4}$)-FF any Hamiltonian.  

Given an $n$ qubit Hamiltonian $H$, one can define an $m=(3n)^c$ qubit Hamiltonian $H'=H\otimes \mathbbm{1}_{2^{m-n}}$. $H'$ is still 2-sparse row computable, therefore it can be FF using the weaker procedure for $T=2^{(m^{1/c})}\ge 5^n$ 
and with $\alpha = m^{-4/c}=(3n)^{-4}<n^{-4}$ (polynomial complexity in $m$ is also polynomial in $n$). Hence \texttt{OEOTL} can be  solved efficiently using the weaker generic FF procedure.
$\square$

\subsection{Does BQP=PSPACE imply generic fast forwarding?} \label{sec:BQPinPSPASEimplyallFF}
For the computational complexity point of view, 
it is interesting to ask whether
the converse of Theorem \ref{thm:No2sparseFF} holds. 
Explicitly, is it true that BQP=PSPACE 
implies that all physically realizable Hamiltonians can be FF? 
One might hope to prove this along the following lines.  
\cite{BCCKS14} show that 
$d$-sparse row-computable Hamiltonians, for polynomial $d$ 
(For our purposes, we can think of those as the 
physically realizable Hamiltonians)    
 can be simulated in quantum 
polynomial time to within exponentially good accuracy. 
Using a quantum PSPACE machine, 
such a Hamiltonian can therefore be applied exponentially many 
times on a given input state, whereas the accumulated 
error is kept bounded. Thus, these Hamiltonians can be fast forwarded in 
quantum polynomial SPACE. One might hope to use the assumption that 
$\BQP=\PSPACE$ to complete the argument; indeed, 
it is known that the class quantum 
polynomial space is equal to PSPACE \cite{Watrous09} which is equal to BQP by assumption.  
However there is a serious difficulty in turning this into a proof: 
we do not know that equalities between classes of decision 
problems could be translated to equivalences in the ability 
to simulate unitary evolutions. 
As pointed to us by Aaronson \cite{AaronsonPrivate}, this is tightly related to the 
{\it unitary synthesis problem} defined in \cite{AK07}. The converse of Theorem 
 \ref{thm:No2sparseFF} 
 thus remains an open problem. 

\section{Quantum algorithms and fast-forwarding Hamiltonians} \label{sec:QAlgAndFF}
One can ask a conceptual question: 
is fast-forwarding Hamiltonians the true underlying
source for all quantum algorithmic speed-ups? 
It turns out that in fact this is far from being the case. Indeed, like in Shor's algorithm, 
the Abelian
hidden subgroup problem (HSP) is solved \cite{Kitaev95,ME99} by efficiently  utilizing phase estimation to exponential accuracy, thus one can associate a Hamiltonian to the problem, and the quantum algorithm can be translated to a cTEUP violation in measuring the energies with respect to this Hamiltonian. We believe (though we have not worked out the details) that this is also the case for the recent extensions of Shor's algorithm to finding unit groups of number fields 
\cite{hallgren07,EHKS14}, which are also based on phase estimation of the eigenvalue of a unitary applied to exponential powers.
However, to our current understanding, other than these few direct extensions 
of Shor's algorithm, none of the other known quantum algorithmic 
speed-ups can be related to fast forwarding -- 
not even quadratic fast forwarding.  We note that some of these algorithms 
can be viewed an an energy measurement of a corresponding Hamiltonian, as 
we describe below, however, the quantum speed-up does not result 
from a FF of this Hamiltonian. We describe this in three interesting cases. 

\begin{enumerate} 
\item \textbf{The exponential speed-up of the 
quantum walk on two glued binary trees \cite{CCDFGS03}:}  
As shown in \cite{CCDFGS03}, the glued trees problem is 
highly symmetric, and 
the search is limited to a subspace of dimension {\it linear} 
in the number of qubits. In addition, 
\cite{CCDFGS03} show 
that the spectral gap of the Hamiltonian in that subspace 
is inverse polynomial. \\
One can in fact view this process as an energy measurement, except not 
an accurate one. To see how  
continuous time quantum walks (CTQW) are related to energy measurements, 
consider the following analogy: 
In CTQW, a value $t$ is chosen uniformly over $[0,T]$ and the system is 
evolved by $e^{-iHt}$ and then measured. Almost equivalently, one can add to 
the 
state an ancilla register, initiated in the 
 superposition over all values of time 
$\frac{1}{\sqrt{T}}\sum_{t=0}^{T-1} \ket{t}$, and then apply 
the Hamiltonian on the state for a duration 
$t$ conditioned that the value in the ancilla register is $t$, 
and finally discard the $t$ register. This latter procedure is effectively 
a phase estimation (i.e., energy measurement), with 
the outcome traced out.  \\
However, the algorithm in \cite{CCDFGS03}
only requires polynomial accuracy to perform this energy measurement, and 
in order to do this 
it simply applies the Hamiltonian for a polynomial amount of time, 
and does not utilize any fast-forwarding 
(equivalently, it does not violate the cTEUP). 
\item \textbf{Grover's quadratic algorithmic speed-up \cite{Grover96}: }
In the Grover's algorithm, an initial state $\ket{s}$ which is a uniform superposition over a search space of size $N$ is rotated slowly to the marked state $\omega$, and reaches its proximity after $O(N^{-1/2})$ applications of the iterator $U=(\mathbbm{1}-2\ketbra{\omega}{\omega})(2\ketbra{s}{s}-\mathbbm{1})$. 
$U$ may be written as:
\begin{equation}
U=(\mathbbm 1 - 2\ketbra{\omega}{\omega}) \prnt{\frac{2}{N} \prnt{(N-1)\ketbra{s'}{s'}+\ketbra{\omega}{\omega}+\sqrt{N-1}(\ketbra{s'}{\omega}+\ketbra{\omega}{s'})} -\mathbbm{1}}
\end{equation}
where $\ket{s}=\sqrt{(N-1)/N}\ket{s'}+\sqrt{1/N}\ket{\omega}$. The subspace spanned by $s',\omega$ is invariant to $U$; by denoting $\ket{\omega}=\ket{0}$ and $\ket{s'}=\ket{1}$, 
\begin{equation}\label{eq:Ugrover}
\begin{split}
U&=\frac{1}{N}\left(\begin{array}{cc}
N-2 & -2\sqrt{N-1}\\
2\sqrt{N-1} & N-2
\end{array}\right)=\mathbbm 1\cos\prnt{\frac{2\sqrt{N-1}}{N}}-i\sin\prnt{\frac{2\sqrt{N-1}}{N}}\sigma^{y}+O(N^{-3/2})\\
&=e^{-2i\frac{\sqrt{N-1}}{N}\sigma^{y}}+O(N^{-3/2}).
\end{split}
\end{equation}
Here we used the following:
\begin{gather}
e^{i\varphi \sigma^y}=\sum_{j=0,1,2...} \frac{\prnt{i\varphi \sigma^y}^j}{j!}= \sum_{j=0,2,4,...} \frac{\prnt{i\varphi}^j}{j!}\cdot \mathbbm{1}+\sum_{j=1,3,5,...} \frac{\prnt{i\varphi}^j}{j!}\cdot \sigma^y=\mathbbm{1}\cos(\varphi)+i\sigma^y \sin(\varphi) \\ 
1-\frac{2}{N} + O(N^{-2})= \cos \prnt{\frac{2\sqrt{N-1}}{N}}\\
\frac{2\sqrt{N-1}}{N} + O(N^{-3/2})= \sin\prnt{\frac{2\sqrt{N-1}}{N}}
\end{gather}
Denote $H= 2\sigma^y/\sqrt N$; then $H$ has
eigenstates $\frac{1}{\sqrt 2}(\ket{s'}\pm i \ket{\omega})$, and additionally,
\begin{equation}
\norm{e^{-iH}-U}=\norm{e^{-2i\sigma^y/\sqrt N}-e^{-2i\sigma^y\sqrt{N-1}/N}+O(N^{-3/2})}=O(N^{-1})
\end{equation}
Measuring an eigenstate of $H$ in the original standard basis returns $\omega$ with probability half. Thus an algorithm equivalent to Grover's is 
to apply an energy measurement of the state $s$ with respect to the 
Hamiltonian $H$, with sufficient accuracy to arrive at a state close to an 
eigenstate, and 
then to measure in the original standard basis. 
Since the two eigenvalues differ by $\theta(\frac{1}{\sqrt{N}})$, it turns out 
that it suffices to perform a measurement with $\eta$-accuracy $N^{-1/2}/10$ 
for $\eta=1-10^{-3}$ to achieve probability at least $1/3$ to measure 
$\omega$. 
The exact argument follows from similar arguments to those in the 
proofs of claims \ref{cl:ExpectationOfDominantEigvec} and \ref{cl:half}\footnote{Let $\varepsilon_j,g$ be the energy measurement output and the state of the garbage register respectively, and let $a_j$ be the amplitude of the eigenstate with energy closest  to $\varepsilon_j$. According to claim \ref{cl:ExpectationOfDominantEigvec}, $\mathbbm{E}_{j,g} |a_j|^2 \ge \eta$ . Similarly to claim \ref{cl:half}, given $j,g$ the marked state is found with probability $\Pr(\omega | j,g)=|a_j|^2 (\frac{1+99\sqrt 2}{2})-\frac{99}{\sqrt 2} -0.01$, and thus $\Pr(\omega)\ge \eta (\frac{1+99\sqrt 2}{2})-\frac{99}{\sqrt 2} -0.01$.}. 

The energy measurement with such accuracy is implemented by phase estimation of $U$ in complexity $O(N^{-1/2})$ by choosing $\ell=b+10$, $b=\ceil{4+\frac{1}{2}\log N}$ 
in Lemma \ref{lem:PEConfidence}. Note that $U$ only approximates $e^{-iH}$.
The phase estimation circuit (figure \ref{fig:ShorPE}) uses  $O(\sqrt N)$ instances of $U$, hence the state after the circuit (before the collapse) deviates by $O(\sqrt{N}\norm{U-e^{-iH}})=O(N^{-1/2})$ compared to phase estimation of $e^{-iH}$. However to the probability to measure $\omega$ after this approximation is still greater than $1/3$ in these parameters.

No FF is required for the speedup (compared to a classical algorithm), as this phase estimation only applies $H$ for time durations 
which are at most $O(\sqrt{N})$. 
The quadratic speed-up is achieved by the mere fact that the accuracy required to separate the two eigenstates is of the order of $1/\sqrt{N}$ and not $1/N$.
\item \textbf{Exponentially fast solutions of 
linear equations \cite{HHL09,CKS15}:} Given an $N\times N$ Hermitian $s$-row computable matrix $A$, and a state $\ket{b}$, the algorithm \cite{HHL09} finds the state $\ket{x}=\sum_i x_i \ket{i}$ for $x$ that solves the equation $Ax=b$ 
(the non Hermitian case can be easily reduced to the Hermitian case). 
The time complexity of the algorithm is $O(poly(\log(N),\kappa, 1/\epsilon)$, 
where $\kappa$ is the condition number of $A$, i.e., the ratio between the 
largest and smallest eigenvalues of $A$, and $\epsilon$ is the additive 
error of $\ket{x}$ allowed. The heart of the algorithm is a phase estimation 
of the unitary matrix $e^{iA}$ applied to the state $\ket{b}$.  
The Hamiltonian simulation procedures used to simulate $e^{-iAt}$ in
\cite{HHL09,CKS15} apply for any $A$, thus both require at least linear 
computational complexity in $t$. If it weren't so, one could violate cTEUP for unknown Hamiltonians - contradicting 
Theorem  \ref{thm:cTEUPUnknown}.  Hence no fast forwarding is involved.

\end{enumerate} 

As for other famous  quantum algorithmic speed-ups, these do not seem to have a sensible description in terms of energy measurements of associated Hamiltonians, so they also do not seem to be related to FF. In particular,
Kuperberg's sub-exponential algorithm for finding 
a hidden subgroup of the Dihedral group \cite{kuperberg05} 
and $\BQP$-complete Topological Quantum Field Theory (TQFT) based quantum algorithms 
\cite{FKW02, BFLW05b, AJL09}, do not seem to have a FF origin.

\subsection{A note on Graph Automorphism}\label{sec:GI}

We consider two problems related to the symmetric group, or the group
 of permutations of $n$ elements, 
denoted $S_n$. 

\begin{definition} {\bf Permutation orbit}\label{def:PO}: 

\noindent$\mathtt{Input:} $
Two permutations
$\sigma, \tau \in S_n$.

\noindent$ \mathtt{Output:} $
Is there a $k$ such that $\sigma^k=\tau$ (i.e., does the subgroup 
of $S_n$ generated by $\sigma$ contain $\tau$?).   
\end{definition} 

The above problem is phrased entirely in group theoretical language, 
and it turns out that it has a $poly(n)$ {\it classical} solution 
\footnote{\label{foot:ChineseRT}
\label{footnoteGA} We're to find the minimal $k$ s.t. $\tau=\sigma^k$. For any $j=1..n$ we can exhaustively find the minimal $0<k_j< n$  s.t. $\tau(j)=\sigma^{k_j} (j)$ (if no solution is found for any $k_j$, then $\tau$ is not a power of $\sigma$). Furthermore, $\tau(j)=\sigma^{k_j+m r_j} (j)$ for every $0<m\in \mathbb{N}$, where $r_j$ is the size of the cycle in $\sigma$ containing $j$.  The problem reduces to finding $k$ satisfying a system of linear congruences $k\equiv k_j \mod r_j$. Solving such system is done by an efficient classical algorithm \cite{Rosen05}. 
}.

One can define a variant of this problem in which 
the permutation group acts on the set of graphs rather than on itself. 
We call this the Cyclic Graph Automorphism problem (CGA)

\begin{definition}{\bf Cyclic Graph Automorphism}

\noindent$\mathtt{Input:} $
An $n$ vertex graph $\Gamma$, and a permutation
$\sigma \in S_n$.

\noindent$ \mathtt{Output:} $
Is there an automorphism $\tau$ of $\Gamma$ 
s.t. $\tau=\sigma^k$ for some integer $k$. (i.e., does the subgroup 
of $S_n$ generated by $\sigma$ contain an automorphism of $\Gamma$?).   
\end{definition} 

We recall that a permutation $\tau$ is an automorphism of a
graph $\Gamma(V,E)$ if 
\begin{equation}
(v,u)\in E \Leftrightarrow (\tau(v)),\tau(u))\in E
\end{equation}

This is a highly restricted version of the notoriously challenging  
graph automorphism problem (GA) (The general version, GA, is known 
to be at least as hard 
as the graph isomorphism problem  \cite{KST+94}, 
which is a long standing open problem  for quantum algorithms; 
Though a recent breakthrough 
by Babai \cite{Babai16} gave a quasi-polynomial classical algorithm 
for the problem.)  

We observe that unlike the Permutation Orbit problem 
(Definition \ref{def:PO}), the related CGA problem seems like a 
hard problem for classical computers.   
In particular, the size of the largest order of an element of the symmetric group $S_n$ is given by Landau's function $g(n)$, which is super polynomial:
\begin{equation}
\lim _{n\rightarrow\infty} \frac{\ln(g(n))}{\sqrt{n\ln (n)}}=1. 
\end{equation}
Hence, the CGA problem cannot be solved in polynomial time
by a brute-force search. Moreover, 
the classical algorithm for the group-theoretical variant of the problem, 
given in Footnote \ref{footnoteGA}, 
cannot be applied to CGA since $\tau$ is not given explicitly.

The following claim is easy to see: 
\begin{claim} 
An efficient quantum algorithm exists for 
the CGA problem. 
\end{claim} 
The proof is essentially by finding the length of the orbit as
in Shor's algorithm; this is due to the ability to fast forward the 
Hamiltonian corresponding to the application of $\sigma$ on the given 
graph. 

\begin{proof} 
Let $K$ be a cyclic subgroup of $S_n$ of order at most $g(n)$, 
generated by $\sigma$. Let $U_\sigma$ be the unitary  
that applies the group action 
of $\sigma$ on a graph $\Gamma$ (represented here by its adjacency matrix):
\begin{equation} \label{eq:permutationOnGraph}
U_\sigma\ket{\Gamma,\Gamma} = \ket{\Gamma,\sigma \Gamma \sigma^{-1}}
\end{equation}
If there is no automorphism of $\Gamma$ in $K$ 
then the orbit of $\Gamma$ under the subgroup generated by 
$\sigma$ is of size equal to the order of $\sigma$ (namely
the minimal $r$ such that $\sigma^r$ is the identity). 
Otherwise this orbit is of size $m<r$ s.t. $m$ is the minimal positive integer 
satisfying $\sigma^m=\tau$ for $\tau$ at automorphism.   
The order of $\sigma$, $r$, can be calculated classically efficiently\footnote{Let $\sigma=\prod_i \sigma_i$, where $\sigma_i$ are the disjoint permutation cycles composing $\sigma$. Let $r_i$ be the order of $\sigma_i$ ($r_i$ is simply 
the length of the $i\th$ cycle). The order of $\sigma$, $r=\mathrm{lcm}(r_1,r_2,..)$}.
Suffices then to compute the order of $\Gamma$ under the action of $U_{\sigma}$ 
and compare it to $r$. This can be done using phase estimation 
for the unitary $U_{\sigma}$ applied with the initial state 
$\ket{\Gamma,\Gamma}$, exactly as described in Section \ref{sec:ShorAsCounterexample}, 
using the fact that the permutation which is derived by taking an 
exponentially large power of $\sigma$ can be calculated efficiently
classically  (by taking exponential power of each permutation cycle composing $\sigma$).  
\end{proof}

To the best of our knowledge, 
the fact that quantum computation is useful in {\it any} context of GI, albeit
highly restricted, was not noticed before; This is probably because 
the Permutation Orbit problem, namely 
the group theoretical version of the CGA problem, is easy classically, 
and hence the problem was not considered. 

Note that if we randomly pick $\sigma\in S_n$, 
the probability that $\tau\in K$ is at most
 $\frac{g(n)}{n!}\approx \frac{e^{\sqrt {n\ln n}}}{n!}=O(e^{-n\ln n +n +\sqrt{n\ln n}-\frac{\ln n}{2}})$, which means that 
the above algorithm cannot be used to provide a algorithmic speed up 
for the GA or graph isomorphism problem.  

We remark that the above algorithm 
can be easily extended to any Abelian subgroup of $S_n$ 
(the largest such group is of size $3^{n/3}$)
using the algorithm for finding Abelian hidden subgroup \cite{ME99}.

\section{Related work} \label{sec:related}

\subsection{Quantum metrology and the Heisenberg limit}
\label{sec:HeisenbergInMetrology}

The question we study here, of efficiently measuring  
energy of eigenstates, is different, though related, to the question 
of {\it distinguishing Hamiltonians} studied in \cite{CPR00},
or to the related question of {\it sensing} in quantum metrology \cite{GLM06, GLM11}. 

A typical problem in quantum metrology is to estimate 
a parameter $\gamma$ of a Hamiltonian $H=\gamma H_0$ 
that resides in a black box, where $H_0$ is a known dimensionless Hamiltonian. 
Consider a Mach-Zehnder interferometer  aimed to determine the phase $\gamma$ of a phase shifter (i.e., $H=\gamma$, and the probe, namely the photon, 
 passes through the box in one time unit.). 
Assuming the optical paths are equal, the Fock state of the probe before the measurement is $\cos(\gamma/2) \ket{1}_a\ket{0}_b+i \sin(\gamma/2) \ket{0}_a\ket{1}_b$, where $a,b$ denote the two spatial modes. $\gamma$ is 
estimated by measuring which of the paths the probe had taken (and of course 
taking statistics over many experiments).  
Classically, one can improve the accuracy (in this context, the standard deviation) by repeating the experiment $n$ times. Assuming no correlation between probes, the accuracy is improved by a factor of $\sqrt n$. This is known as the \emph{standard quantum limit}. When quantum correlations are introduced, the accuracy can bypass the standard quantum limit \cite{Yurke86}, by using $n$ photons simultaneously, instead of repeating a single photon experiment $n$ times. For instance, the system is in the a ``Noon state'' \cite{Sanders89,LKD02}
 $\ket{n}_a\ket{0}_b+\ket{0}_a\ket{n}_n$, before entering the phase shifter. 
Each probe then gets a global phase factor as it passes through the phase shifter, and the state just before the measurement becomes 
$\cos(n\gamma/2)\ket{n}_a\ket{0}_b+i\sin(n\gamma/2)\ket{0}_a\ket{n}_b$. 
The distinguishability compared to a single probe is improved by a 
factor of $n$. 
When the Hamiltonian is accessed as a black box, this scaling of accuracy 
as $1/n$ is optimal; it is referred to as the the 
\emph{Heisenberg limit}\cite{GLM06,GLM11,ZpDK10,ZpDK12}, where $n$ is the number of times the Hamiltonian was applied (for one time unit each)
\footnote{A careful examination shows that no entanglement is necessary if one is allowed to modify the apparatus: one can direct a single photon $n$ times through the  phase shifter, and get exactly the same distribution as in the Noon state experiment}.
Such limits seem related to the TEUP; In particular, 
the metrology setting 
considers Hamiltonians 
which are promised to be of the form $H=\gamma H_0$, 
where $H_0$ is known and $\gamma$ is unknown; 
Inaccuracy in $\gamma$ is thus directly related to inaccuracy in the 
energy eigenvalues. Moreover, the metrology setting allows accessing the 
Hamiltonian only through a black box. 
This is very reminiscent of the conditions for TEUP with
unknown eigenvalues.  
It is thus natural to ask 
whether the Heisenberg limit is equivalent in some sense to the TEUP 
for Hamiltonians with unknown eigenvalues, Theorem \ref{thm:TEUP23confidence}. 
Indeed, it seems that an equivalence between the two should be provable, 
but we have not worked out the details. Roughly, the argument is as follows.  

To derive a TEUP from the Heisenberg limit,  
assume we are given (or can generate) 
an eigenstate of $H_0$ with eigenvalue $E$, 
and also assume we can apply an energy measurement 
with standard error $std(E)$ on this eigenstate. 
This implies a standard error $std(E)/E$ 
of $\gamma$; However in this setting we can apply the Heisenberg limit 
(it holds even though we assume infinite computational power outside 
of the black boxes, so we can assume we can generate the eigenstate). 
This sets an upper 
bound on the accuracy to estimate $\gamma$, and so it 
entails a bound on the accuracy of the energy measurement. 
The details of this implications remain to be filled in, 
since the two settings use 
different accuracy and resource measures.  

The other direction, of deriving a version of the Heisenberg limit from the 
TEUP, is a little less clear. 
Indeed, one can reduce the energy measurement problem to the 
problem of estimating $\gamma$, 
by applying the transformation  $\ket{\psi_E,0}\mapsto\ket{\psi_E,\gamma' E}$, 
with $\gamma'$ being the estimated value of $\gamma$. 
However, we do not have a TEUP for the case of Hamiltonians that are 
completely known up to a scalar factor $\gamma$; it seems that such 
a TEUP should follow from arguments similar to those in \cite{AMP02},
but this remains to be done.

\subsection{Other quantum bounds}
\subsubsection{Hamiltonian in a black box}
%
%

We review two bounds on Hamiltonian dynamics.

\emph{The Mandelstam-Tamm relation} \cite{MT45,Fleming73,Bhattacharyya83}, which was also named ``the time energy uncertainty principle'', is the following:

\begin{equation}
\norm{\bra{\psi} e^{-iHt} \ket{\psi}}^2 \ge \cos ^2 \prnt{\frac{\Delta E t}{\hbar}} ~~~~ 0\le t \le \pi \hbar / 2\Delta E,
\end{equation}
and $(\Delta E)^2 = \bra{\psi} H^2 \ket{\psi} - \bra{\psi} H \ket{\psi}^2$. Namely, the change of the quantum state's direction, when evolving under the 
Hamiltonian $H$, is bounded from below 
in terms of its energy variance with respect to eigenvalues of $H$.   

Let $H$ be a Hamiltonian with ground energy $0$. 
The \emph{Margolus-Levitin bound} on the time duration required for a state $\psi$ to evolve under $H$ to a state orthogonal to $\psi$ is the following:
\begin{equation}
\tau_{\perp} \ge \frac{\pi \hbar}{2 \bra{\psi}H\ket{\psi}}
\end{equation}
Here the state's evolution is bounded by the expectation value of the energy. 

Both bounds refer to the dynamics of a given Hamiltonian, which may not be altered (in a similar way to when $H$ is applied as a black box). 
Hence these bounds do not apply to shortcuts such as fast forwarding, 
or simulations of the Hamiltonian by equivalent quantum circuits. 
A relation between the two bounds and the Heisenberg limit is discussed in \cite{ZpDK10}.

\subsubsection{The No-FF Theorem in the Query Model}\label{sec:Hsim}

A different model of Hamiltonian simulation, called the Hamiltonian 
query model, was studied in  \cite{BACS07,ChildsThesis, CK11,CW12,BCCKS14}. 
In this model, the access to the entries of the row-sparse Hamiltonian is by 
queries to a row-oracle. Given the index of the row, 
the row-oracle returns the column numbers and the corresponding values of all non-zero elements in the row. 
Hamiltonian simulation algorithms in this model \cite{BACS07,ChildsThesis, CK11,CW12,BCCKS14},
 minimize the number of queries. 
 Unlike Hamiltonians in which the circuit simulating the Hamiltonian is known 
(which is the subject of most of this paper)  
in the query model the circuit applying the Hamiltonian is {\it hidden} 
and therefore the only way to gain information about the Hamiltonian is
by queries. 

The most recent algorithm \cite{BCCKS14} for  simulating a $d$-sparse normalized Hamiltonian $H$ for time $t$ with additive error $\varepsilon$ in this query 
model 
requires  $O(\frac{\tau \log (\tau/\varepsilon)}{\log\log(\tau/\varepsilon)})$ queries, with $\tau=d^2t$. 
Note that the query complexity is more than linear in the time of the 
simulation, and it depends logarithmically on $\varepsilon^{-1}$. It 
is independent of the number of qubits.

In \cite{BACS07}, it was shown that this behavior on time is essentially 
optimal: 
no procedure could simulate Hamiltonians in the query model 
for time $t$ with query complexity that is sub-linear in $t$. 
This can be viewed as a no-fast-forwarding theorem for the query model. 
(The proof is by showing that such a procedure implies an algorithm 
for finding the parity of a binary string in less queries than the known 
lower bound \cite{BBCMW98}).
One might wonder whether this theorem can be derived from
the cTEUP for unknown Hamiltonians (Theorem \ref{thm:cTEUPUnknown} adapted 
from \cite{AMP02}) which together with our theorem \ref{thm:main} 
implies a no-FF for unknown Hamiltonians.  
However this result cannot be applied for the query model; 
the reason is that the query model has better distinguishability than the model of unknown Hamiltonians.   
For example, finding an entry of the Hamiltonian with accuracy of $poly(n)$ 
bits costs one query in the query model, 
unlike the exponential time required in the case where 
the Hamiltonian is given as a black box. 

We can summarize the comparison between the three models: 
a Hamiltonian given as a black box or one with unknown eigenvalues can't be 
Fast forwarded as this violates the TEUP/cTEUP for unknown Hamiltonians (Theorems \ref{thm:TEUP23confidence},\ref{thm:cTEUPUnknown} respectively). Adding information on the Hamiltonian by giving access to it by row oracle still won't allow a general FF procedure due to the no-FF theorem relying on 
query complexity bounds \cite{BACS07}. 
Our Theorem \ref{thm:No2sparseFF}
is the corresponding theorem 
for the case of $2$-sparse row computable 
Hamiltonians; Since we are no longer in the black box model,
or even in the query model, 
we must condition the result on computational assumptions 
(here, we rely on the widely believed assumption that  $\PSPACE\ne \BQP$).

\subsection{Susskind's complexificaition of a wormhole's length}
Theorem \ref{thm:No2sparseFF} surprisingly relates to a recent conjecture by Susskind on the length of non-traversable wormholes \cite{Susskind16}. 
 
The AdS/CFT correspondence \cite{Maldacena99} is a conjectured equivalence, 
or duality, between certain quantum gravity theories in anti de-Sitter  (AdS) spacetime and Conformal Field theories (CFTs). The equivalence implies that 
any physical process can be formalized in either theory, and give the same 
predictions. 
An exact transformation between the two theories is yet to be found.

The CFT dual of a non-traversable wormhole is the maximally entangled state, 
which evolves in time under the transformation:
\begin{equation} \label{eq:SusskindPsit}
\ket{\psi_t}=2^{-n/2} \sum_{y=1}^{2^n} \ket{y}\otimes U^t \ket{y}   ~~~~~ U=V^{-1}V^T,
\end{equation}  
where $V$ is a unitary. 
Susskind \cite{Susskind16} has recently proposed the very intriguing 
suggestion, that the CFT dual of the length of non-traversable wormholes is equal to the {\it quantum circuit complexity} required to approximate $\psi_t$. 
This proposition is currently the only one producing the expected behaviour of non-traversable wormhole length in cases of interest. For instance, quantum gravity  predicts that the length of the wormhole grows linearly in time \cite{FW62} until a bound of $2^n poly(n)$ is reached, and shrinks at time $\approx 2^{2^n}$. In terms of state complexity, 
by time $t< 2^{2^n}$, $\psi_t\approx \psi_0$ as can be proven by a counting 
argument; whereas $2^n poly(n)$ is the upper bound on the complexity for 
approximating a state with constant error (generate one amplitude of a standard basis state at a time). 
A natural question is whether there exists a unitary $U$ which achieves this 
upper bound, or in more picturesque words, {\it complexifies} 
the state $\psi_t$ 
so rapidly. 

Aaronson and Susskind \cite{AS16p,Aaronson16} do not handle 
the particular $U$ of the CFT, but prove that there {\it exists} 
a unitary with almost maximal complexification 
(more precisely, that $\psi_t$ with this $U$ cannot
be approximated efficiently for some $t<2^n$). Their proof works   
under a commonly believed computational assumption 
(\PSPACE $\nsubseteq$\PP /poly). 
In their terminology, they show that 
there are no ``shortcuts'' to generating the state $\psi_t$ for such a $U$.   
This ties with our no-generic FF Theorem \ref{thm:No2sparseFF}: 
note that if the Hamiltonian $H$ generating 
the unitary $U$, s.t. $U=e^{-iH}$, could be exponentially 
{\it fast-forwarded}, the state complexity of $\psi_t$ would by polynomial. 
Thus, impossibility of FF of $H$ follows from impossibility to generate 
$\psi_t$ efficiently. 
The other way round might not hold - it is conceivable that 
simulating the related unitary, namely  fast forwarding, is impossible, 
but simulating the {\it state} $\psi_t$ can be generated efficiently 
somehow by a different way. 
This is why the computational assumption in Aaronson and Susskind's result is 
stronger than ours, and involves the class $\PP$ and not $\BQP$. 
The two other differences between the two theorems
(they work in the non-uniform setting, namely use $\PP/\mathrm{poly}$ rather than $\PP$, 
and consider approximation of the state to within a constant), 
depend on the setting and are less important.

We conclude that the notion of fast forwarding and limitations on it 
seem relevant also in the context of 
questions arising in quantum gravity.

\section{Conclusions and Open Questions}\label{sec:conc}
A fundamental open question remains:  
Is it possible to characterize, by Physical terms, the exact 
conditions for fast forwarding? Or equivalently, for super-efficient 
energy measurements? What is the true physical reason for such 
a possibility? 

A very intriguing question in this context is whether 
many body localization Hamiltonians, which had attracted much attention 
recently \cite{NH14}, and which are 
in some sense a generalization of commuting local Hamiltonians, 
can be fast forwarded efficiently; by Theorem \ref{thm:main}
this would imply the ability to measure 
the energy of such systems with exponential accuracy. 

Relatedly, it is possible that 
the approach taken in this work  
might help in resolving lead the following disturbing question: 
Why is it that though the TEUP can easily be violated, 
it does seem to be satisfied in many ``typical'' physical systems? 
We speculate that the reason for that is that 
a typical local Hamiltonian cannot be fast forwarded. 
By our Theorem \ref{thm:main}, such a Hamiltonian would obey the cTEUP 
(as well as the weaker TEUP). 
Proving that a typical local Hamiltonian cannot be fast forwarded, 
would then constitute an explanation to the observed phenomenon that the  
TEUP is obeyed in most cases.  
One way towards such a proof would be to try to mimick the proof 
of Theorem \ref{thm:No2sparseFF}
for a typical, randomly chosen local Hamiltonian. 
The claim seems likely to be true; a possible way to prove it might be to show it to be true for every universal Hamiltonian, defined in an appropriate sense, and then show that a typical local Hamiltonian is universal in that sense. 

As discussed in Subsection \ref{sec:BQPinPSPASEimplyallFF}, 
it remains open whether the converse of Theorem \ref{thm:No2sparseFF} holds 
or not. 
The relation between quantum algorithms and fast forwarded Hamiltonian 
put forward in this work also raises
another question motivated by computer science - 
what other interesting problems can have their solution hidden in the 
exponentially accurate spectrum of a Hamiltonian that can be fast forwarded? 
Could finding new examples for fast forwarding 
interesting Hamiltonians, lead to new quantum algorithms? 

More generally, the very notion of fast forwarding Hamiltonians, 
namely, projecting 
the physical system far into the future using computationally 
limited resources, captures the imagination. What else can be 
done with such systems, experimentally or theoretically? 

Last but not least, our work leaves open the question of whether 
parameter estimation in sensing and metrology 
can also be dramatically improved (beyond 
the Heisenberg limit, as in \cite{HGAKR15, KLSL14}). As explained in \ref{sec:HeisenbergInMetrology}, this is impossible to do in the usual sensing paradigm in which the Hamiltonian is given in a black box. However it is conceivable that one can use more knowledge about the Hamiltonian and possibly quantum computational techniques such as fast forwarding or others to go beyond the 
Heisenberg limit.  This possibility is left for future exploration. 

\section{Acknowledgements}
We are grateful to Scott Aaronson, Fernando Brand\~{a}o, Nadav Katz,
Leonid Palterovitch,
Daniel Roberts, Or Sattath and 
Norman Ying Yao, for very useful comments.  
In particular, we acknowledge a discussion with Norman Yao which eventually 
led to Theorem \ref{thm:quadFF}. 
\section*{Appendix}
\appendix

\section{TEUP for unknown Hamiltonians} \label{apndx:TEUPforUnknown}
We slightly strengthen the TEUP for completely unknown Hamiltonians \cite{AMP02}, by showing it holds even when the eigenstates are known. Then we  adapt it to comply with $\nicefrac{2}{3}-$accuracy (definition \ref{def:eta-accuracy}) instead of mean deviation.
In this section, our definition of a Hamiltonian with unknown eigenvalues is as 
stated in Theorem \ref{thm:cTEUPUnknown}. 

\begin{theorem} [TEUP for Hamiltonians with 
unknown eigenvalues (adapted from  \cite{AMP02})] \label{thm:AMPteup}
If the eigenvalues of a Hamiltonian acting on a system are
unknown, then the precision $\Delta E$ with which one can estimate
the energy of an eigenstate with energy $E$ in a time  obeys the constraint
\begin{equation}
\Delta E\Delta t > 1/4.
\end{equation}
$\Delta E$ is the mean deviation of the measurement:
\begin{equation} \label{eq:MeanDeviation}
\Delta E=\sum_{E'} \Pr (E' | E) \abs{ E-E'}
\end{equation}
$\Delta t$ is the total time the Hamiltonian was applied.
\end{theorem}
\begin{proof}
The proof of TEUP for completely unknown Hamiltonians\cite{AMP02} holds for the case in which only eigenvalues are unknown, with $\Delta t$ denoting the duration of the measurement, and the scheme is sequential, i.e.,  only one system is affected by the Hamiltonian at any given time. 
One might be worried however that by applying the Hamiltonian on several 
probes in parallel on entangled states as is done for example 
in the case of Noon 
states (see Section \ref{sec:HeisenbergInMetrology}), 
one might be able to  bypass the bound achieved in 
\cite{AMP02}. 
However, notice that given 
any measurement scheme which applies the Hamiltonian 
on several probes (registers) in parallel,
one can apply standard quantum computation techniques of adding a register and swapping between registers, to arrive at an equivalent 
protocol which only applies the Hamiltonian on one probe (register) 
sequentially. Hence  we take $\Delta t$ to be 
the total time the Hamiltonian was applied. 
\end{proof}

In terms of $\nicefrac{2}{3}$-accuracy, Theorem \ref{thm:AMPteup} takes the following form:
\begin{theorem} \label{thm:TEUP23confidence}
Let $H$ be a Hamiltonian with unknown eigenvalues. The $\nicefrac{2}{3}$-accuracy $\delta E$ of measuring the energy of an eigenstate depends on the time the Hamiltonian was sampled $\Delta t$ by
\begin{equation}
\delta E \Delta t \ge 1/3
\end{equation}
\end{theorem}

\begin{proof}
Assume by contradiction that there exists 
some family of Hamiltonians with unknown eigenvalues (but fixed eigenstates
which are common to all),   
and also that there exists a given eigenstate and constants 
$\delta E$ and $\Delta t$, s.t.
one can perform an energy measurement of the eigenstate with 
$\nicefrac{2}{3}$-accuracy $\delta E$, 
while applying the Hamiltonian for 
$\Delta t$, and yet $\delta E \Delta t<\frac{1}{3}$.  
We will derive a contradiction by showing that this implies a protocol 
which is too strong, for the distinguishability problem studied in 
\cite{AMP02}.  

The distinguishability problem is defined as follows: 
Given access to a Hamiltonian by a black box, determine whether the 
Hamiltonian in the box is 
$H_1$ or $H_2=H_1+\varepsilon\mathbbm{1}$ (it is promised 
that the Hamiltonian in the box is one of the two, and it is assumed that 
there are no computational bounds outside the box, and in particular, 
we can feed the box any eigenstate we want). 
Both Hamiltonians have an a-priory 
probability $1/2$. Define the probability of error for a protocol 
for this task by: 
\begin{equation}
P_{err} = \frac{1}{2} 
\prntt{\Pr(output~2|H_1)+\Pr(output~1|H_2)}. 
\end{equation} 

Using our assumed energy measurement, we can derive a protocol for 
this distinguishability task between two Hamiltonians from the family: 
$H_1$ and $H_2=H_1+\varepsilon\mathbbm{1}$, with 
$\varepsilon=\nicefrac{2}{3\Delta t}$.  
Apply an energy measurement with $\nicefrac{2}{3}$-accuracy 
$\delta E<\nicefrac{1}{3\Delta t}=\nicefrac{\varepsilon}{2}$, 
to an eigenstate of the Hamiltonians (which by assumption 
we can generate). We know the energy of the eigenstate 
is either $E$ or $E+\varepsilon$.    
The procedure outputs $H_1$ if the measurement outcome is closer to $E$ than to $E+\varepsilon$ and outputs $H_2$ otherwise.  
From the definition of $\nicefrac{2}{3}$-accuracy, in this procedure,
\begin{equation} \label{eq:DeltaHcounter}
P_{err} < 1/3.
\end{equation}

However, one of the intermediate results on Hamiltonian 
distinguishability in \cite{AMP02} is the following: 
\begin{lemma} [$H$ distinguishability, adapted from \cite{AMP02} section III.B
]
Any algorithm solving 
the distinguishability problem defined above for distinguishing between 
$H_1$ and $H_2=H_1+\varepsilon\mathbbm{1}$, 
while applying the Hamiltonian in the black box for a total time $\Delta t$, 
satisfies
\begin{equation} \label{eq:DeltaH}
P_{err} \ge 
\frac{1}{2} \prntt{1-\sin\prnt{\frac{\varepsilon \Delta t}{2}}} 
\end{equation} 
if $\varepsilon \Delta t < \pi$. 
\end{lemma}

Combining this lemma and Equation \ref{eq:DeltaHcounter} we have 
\begin{equation}
\frac{1}{3} > P_{err} \ge \frac{1}{2} \prntt{1-\sin\prnt{\frac{1}{3}}}, 
\end{equation}
which is a contradiction. 
\end{proof}

\section{Proofs of claims in Section \ref{sec:quadratic}}
\setcounter{claim}{2}
\begin{claim}
Let $\textbf{a}$ be a
column vector whose $j\th$ coordinate is $a_j$ and let $\textbf{a}^\dagger$ be  a column vector whose $j\th$ coordinate 
is $a_j^\dagger$. The Hamiltonian $H$ can be written as 
\begin{equation}
H=\frac{1}{2}\left(\begin{array}{cc}
\overline{\mathbf{a}^{\dagger}} & \overline{\mathbf{a}}\end{array}\right)\left(\begin{array}{cc}
A & B^{*}\\
B & -A^{*}
\end{array}\right)\left(\begin{array}{c}
{\mathbf{a}}\\
{\mathbf{a}}^{\dagger}
\end{array}\right)+\frac{1}{2}\tr(A)
\end{equation} 
Here, the overline  indicates a matrix transposition, i.e., $\overline{\mathbf{a}^{\dagger}}, \overline{\mathbf{a}} $ are the row vectors corresponding to $\textbf{a}^\dagger, \textbf{a}$ respectively.
\end{claim}

\begin{proof}
The proof relies on the hermiticity of $A,B$ and the anticommutation relations. For $i\neq j$ 
\begin{equation}
\sum_{i\neq j} A_{i,j}a_i^\dagger a_j= 
\frac{1}{2} \sum_{i\neq j} \prnt{A_{i,j}a_i^\dagger a_j - A_{i,j}a_j a_i^\dagger} =
\frac{1}{2} \prnt{\sum_{i\neq j} A_{i,j}a_i^\dagger a_j - \sum_{j\neq i} A_{j,i}a_i a_j^\dagger}=
 \frac{1}{2} \sum_{i\neq j}(A_{i,j}a_i^\dagger a_j - A^*_{i,j}a_i a_j^\dagger)
\end{equation}
For $i=j$,
\begin{equation}
\sum_i A_{i,i}a_i^\dagger a_i = 
\frac{1}{2}\sum_i A_{i,i}a_i^\dagger a_i + \frac{1}{2} \sum _i A_{i,i}(1-a_i a_i ^\dagger )= \frac{1}{2}\prnt{\sum_i A_{i,i}a_i^\dagger a_i - \sum_i A^*_{i,i}a_i a_i^\dagger  + \tr (A)}
\end{equation}
Reorganizing $H$ as a block matrix concludes the proof. 
\end{proof}

\begin{claim}
The traceless part of the Hamiltonian can be diagonalized: 
\begin{equation}
H=  \frac{1}{2}\left(\begin{array}{cc}
\overline{\mathbf{a}^{\dagger}} & \overline{\mathbf{a}}\end{array}\right)U D U^\dagger \left(\begin{array}{c}
{\mathbf{a}}\\
{\mathbf{a}}^{\dagger}
\end{array}\right)+\frac{1}{2}\tr(A)
\end{equation} 
where $D$ is a real diagonal matrix s.t. $D_{j,j}=-D_{j+n,j+n}$ and $U$ is unitary. Furthermore, there exist matrices $V_1,V_2$ s.t. $U$ is a block matrix in the form $U=\left(\begin{array}{cc}
V_1 & V_2^*\\
V_2 & V_1^*
\end{array}\right) $
\end{claim}
\begin{proof}
Note the following symmetry of the traceless matrix $\mathcal{H}$:
\begin{equation}
\mathcal{H}\equiv \left(\begin{array}{cc}
A & B^{\dagger}\\
B & -A^{*}
\end{array}\right)~~~~\mathcal{H}^\dagger=\mathcal{H}, ~~~~\tau \mathcal{H} \tau = -\mathcal H^*, ~~~~\tau \equiv \left(\begin{array}{cc}
0 & \mathbbm{1}_m\\
\mathbbm{1}_m & 0
\end{array}\right)
\end{equation} 
The symmetry of $\mathcal H$ implies a symmetry on its eigenvectors, which are the column vectors of $U$.
\begin{equation}
\mathcal{H} \left(\begin{array}{c}
{v_1}\\
{v_2}
\end{array}\right) = \lambda \left(\begin{array}{c}
v_1\\
v_2
\end{array}\right) \Rightarrow
\mathcal{H} \left(\begin{array}{c}
{v_2^*}\\
{v_1^*}
\end{array}\right) = -\lambda \left(\begin{array}{c}
v_2^*\\
v_1^*
\end{array}\right)
\end{equation}
Hence if either $v_1\neq v_2^*$ or $v_2\neq v_1^*$, then the eigenvectors come in pairs: $w,\tau w^*$. In the case $v_1=v_2^*$ and $v_2=v_1^*$, we get that $\lambda =0$. Since there is an even number of eigenvectors of the first case, and the dimension of the subspace is $2m$, there's also an even number of eigenvectors of the second case, with eigenvalue $0$. Picking a pair of orthogonal eigenvectors $w_1,w_2$, with eigenvalue 0, we can span the subspace they define using $w_1'=w_1+iw_2$ and $w_2'=w_1-iw_2$, and we get that $w_1'=\tau w_2'^*$. Hence all eigenvectors of $\mathcal H$ come in pairs $w,\tau w^*$. By placing one eigenvector of the $i\th$  pair in column $i$ of $U$ and the other eigenvector in column $m+i$, $U$ takes the desired form. Additionally, the column placement forces  $D_{j,j}=-D_{j+m,j+m}$.
\end{proof}

\begin{claim}
By defining
\begin{equation} \label{eq:DefFermionsBs}
\left(\begin{array}{c}
{\mathbf{b}}\\
{\mathbf{b}}^{\dagger}
\end{array}\right)=
U^\dagger \left(\begin{array}{c}
{\mathbf{a}}\\
{\mathbf{a}}^{\dagger}
\end{array}\right).
\end{equation}
The Hamiltonian takes the form
\begin{equation}  
H= 2\sum_{i=1}^m b_i^\dagger b_i D_{i,i} + \frac{1}{2} \tr (A)
\end{equation}
The new operators obey the anti-commutation relations of fermions. In addition, the eigenvalues of $b_i^\dagger b_i$ either 0 or 1.  
\end{claim}

\begin{proof}
From equations \ref{eq:DefFermionsBs} and the definition of $U$, 
\begin{gather}
b_i=\sum_j (V_1^\dagger)_{i,j} a_j + (V_2^\dagger)_{i,j} {a}_j^\dagger \\
{b}_i^\dagger=\sum_j(\overline V_2)_{i,j} {a}_j + (\overline V_1)_{i,j}{a}^\dagger_j
\end{gather}
Hence $b_i^\dagger$ is indeed the hermitian conjugate of $b_i$. Additionally, one can see that $\left(\begin{array}{cc}
\overline{\mathbf{b}^{\dagger}} & \overline{\mathbf{b}}\end{array}\right)=\left(\begin{array}{cc}
\overline{\mathbf{a}^{\dagger}} & \overline{\mathbf{a}}\end{array}\right)U$.
Next we show that the transformation is canonical, i.e., the anticommutation relations are preserved: 
\begin{equation}
\begin{split}
\prnttt{b_i,b_j}&=\prnttt{\sum_{k=1}^m U_{k,i}^* a_k + \sum_{k=1}^{m} U_{k+m,i}^* a_{k}^\dagger, \sum_{\ell=1}^m U_{\ell,j}^* a_\ell + \sum_{\ell=1}^{m} U_{\ell+m,j}^* a_{\ell}^\dagger}\\
&=\sum_{k=1}^m \prnt{U_{k,i}^* U_{k+m,j}^*+  U_{k+m,i}^* U_{k,j}^*}=\sum_{k=1}^m \prnt{U_{k,i}^* U_{k,j+m}+  U_{k+m,i}^* U_{k+m,j+m}}=0
\end{split}
\end{equation}
(we used the structure of $U$ for the transition of the second line)

\begin{equation}
\begin{split}
\prnttt{b_i,b_j^\dagger}=\prnttt{\sum_{k=1}^m U_{k,i}^* a_k + \sum_{k=1}^{m} U_{k+m,i}^* a_{k}^\dagger, \sum_{\ell=1}^m U_{\ell,j} a_\ell^\dagger + \sum_{\ell=1}^{m} U_{\ell+m,j} a_{\ell}}
=\sum_{k=1}^m \prnt{U_{k,i}^* U_{k,j}+  U_{k+m,i}^* U_{k+m,j}}=\delta_{i,j}
\end{split}
\end{equation}
Finally, the eigenvalues of $b_j^\dagger b_j$ are 0 and 1, because $(b_j^\dagger b_j)^2=b_j^\dagger b_jb_j^\dagger b_j=(1-b_j b_j^\dagger) b_j^\dagger b_j=b_j^\dagger b_j$,
similarly to the number operator $a_j^\dagger a_j$.
\end{proof}

\bibliography{bib}

\begin{thebibliography}{10}

\bibitem{AB61}
Y.~Aharonov and D.~Bohm.
\newblock {Time in the quantum theory and the uncertainty relation for time and
  energy}.
\newblock {\em Physical Review}, 122(5):1649--1658, 1961.

\bibitem{AMP02}
Y.~Aharonov, S.~Massar, and S.~Popescu.
\newblock Measuring energy, estimating hamiltonians, and the time-energy
  uncertainty relation.
\newblock {\em Phys. Rev. A: At., Mol., Opt. Phys.}, 66:052107, 2002.

\bibitem{CPR00}
Andrew~M Childs, John Preskill, and Joseph Renes.
\newblock Quantum information and precision measurement.
\newblock {\em Journal of modern optics}, 47(2-3):155--176, 2000.

\bibitem{Shor94}
Peter~W. Shor.
\newblock Algorithms for quantum computation: discrete logarithms and
  factoring.
\newblock In {\em In Proceedings of the 35th Annual Symposium on Fundamentals
  of Computer Science (FOCS'94)}, volume~35, pages 124--134, 1994.

\bibitem{Peres93}
Asher Peres.
\newblock {\em Quantum Theory: Concepts and Methods}.
\newblock Fundamental Theories of Physics. Kluwer Academik Publishers,
  Dordrecht, 1993.

\bibitem{Busch08}
Paul Busch.
\newblock The time--energy uncertainty relation.
\newblock In {\em Time in quantum mechanics}, pages 73--105. Springer, 2008.

\bibitem{GLM06}
Vittorio Giovannetti, Seth Lloyd, and Lorenzo Maccone.
\newblock Quantum metrology.
\newblock {\em Physical review letters}, 96(1):010401, 2006.

\bibitem{GLM11}
Vittorio Giovannetti, Seth Lloyd, and Lorenzo Maccone.
\newblock Advances in quantum metrology.
\newblock {\em Nature Photonics}, 5(4):222--229, 2011.

\bibitem{ZpDK10}
Marcin Zwierz, Carlos~A P{\'e}rez-Delgado, and Pieter Kok.
\newblock General optimality of the heisenberg limit for quantum metrology.
\newblock {\em Physical review letters}, 105(18):180402, 2010.

\bibitem{ZpDK12}
Marcin Zwierz, Carlos~A P{\'e}rez-Delgado, and Pieter Kok.
\newblock Ultimate limits to quantum metrology and the meaning of the
  heisenberg limit.
\newblock {\em Physical Review A}, 85(4):042112, 2012.

\bibitem{MT45}
L~Mandelstam and Igor Tamm.
\newblock The uncertainty relation between energy and time in nonrelativistic
  quantum mechanics.
\newblock {\em J. Phys.(USSR)}, 9(249):1, 1945.

\bibitem{AS16p}
Scott Aaronson and L.~Susskind.
\newblock In prepation.

\bibitem{Susskind16}
Leonard Susskind.
\newblock Computational complexity and black hole horizons.
\newblock {\em Fortschritte der Physik}, 64(1):24--43, 2016.

\bibitem{BV97}
Ethan Bernstein and Umesh Vazirani.
\newblock Quantum complexity theory$^\dag$.
\newblock {\em SIAM J.\ Comp.}, 26(5):1411--1473, 1997.

\bibitem{NC00}
Michael~A. Nielsen and Isaac~L. Chuang.
\newblock {\em Quantum Computation and Quantum Information}.
\newblock Cambridge University Press, 2000.

\bibitem{BVT80}
Vladimir~B Braginsky, Yuri~I Vorontsov, and Kip~S Thorne.
\newblock Quantum nondemolition measurements.
\newblock {\em Science}, 209(4456):547--557, 1980.

\bibitem{Rot12}
Joseph Rotman.
\newblock {\em An introduction to the theory of groups}, volume 148.
\newblock Springer Science \& Business Media, 2012.

\bibitem{AT03}
Dorit Aharonov and Amnon Ta-Shma.
\newblock Adiabatic quantum state generation and statistical zero knowledge.
\newblock In {\em Proceedings of the thirty-fifth annual ACM symposium on
  Theory of computing}, pages 20--29. ACM, 2003.

\bibitem{Kitaev95}
A~Yu Kitaev.
\newblock Quantum measurements and the abelian stabilizer problem.
\newblock {\em arXiv preprint quant-ph/9511026}, 1995.

\bibitem{PC99}
Victor~Y Pan and Zhao~Q Chen.
\newblock The complexity of the matrix eigenproblem.
\newblock In {\em Proceedings of the thirty-first annual ACM symposium on
  Theory of computing}, pages 507--516. ACM, 1999.

\bibitem{ABBCS15}
Diego Armentano, Carlos Beltr{\'a}n, Peter B{\"u}rgisser, Felipe Cucker, and
  Michael Shub.
\newblock A stable, polynomial-time algorithm for the eigenpair problem.
\newblock {\em arXiv preprint arXiv:1505.03290}, 2015.

\bibitem{Kitaev2003}
A~Yu Kitaev.
\newblock Fault-tolerant quantum computation by anyons.
\newblock {\em Annals of Physics}, 303(1):2--30, 2003.

\bibitem{BJS10}
Michael~J Bremner, Richard Jozsa, and Dan~J Shepherd.
\newblock Classical simulation of commuting quantum computations implies
  collapse of the polynomial hierarchy.
\newblock In {\em Proceedings of the Royal Society of London A: Mathematical,
  Physical and Engineering Sciences}. The Royal Society, 2010.

\bibitem{BMS16}
Michael~J Bremner, Ashley Montanaro, and Dan~J Shepherd.
\newblock Achieving quantum supremacy with sparse and noisy commuting quantum
  computations.
\newblock {\em arXiv preprint arXiv:1610.01808}, 2016.

\bibitem{Anderson58}
P.~W. Anderson.
\newblock Absence of diffusion in certain random lattices.
\newblock {\em Phys. Rev.}, 109:1492--1505, Mar 1958.

\bibitem{BookIZ06}
Claude Itzykson and Jean-Bernard Zuber.
\newblock {\em Quantum field theory}.
\newblock Courier Corporation, 2006.

\bibitem{Bblaizot86}
Jean-Paul Blaizot and Georges Ripka.
\newblock {\em Quantum theory of finite systems}, volume~3.
\newblock Mit Press Cambridge, 1986.

\bibitem{Shchesnovich13}
VS~Shchesnovich.
\newblock The second quantization method for indistinguishable particles
  (lecture notes in physics, ufabc 2010).
\newblock {\em arXiv preprint arXiv:1308.3275}, 2013.

\bibitem{BCCKS14}
Dominic~W Berry, Andrew~M Childs, Richard Cleve, Robin Kothari, and Rolando~D
  Somma.
\newblock Exponential improvement in precision for simulating sparse
  hamiltonians.
\newblock In {\em Proceedings of the 46th Annual ACM Symposium on Theory of
  Computing}, pages 283--292. ACM, 2014.

\bibitem{Papad94}
Christos~H Papadimitriou.
\newblock On the complexity of the parity argument and other inefficient proofs
  of existence.
\newblock {\em Journal of Computer and system Sciences}, 48(3):498--532, 1994.

\bibitem{Watrous09}
John Watrous.
\newblock Quantum computational complexity.
\newblock In {\em Encyclopedia of complexity and systems science}, pages
  7174--7201. Springer, 2009.

\bibitem{AaronsonPrivate}
Scott Aaronson.
\newblock Private communication.

\bibitem{AK07}
Scott Aaronson and Greg Kuperberg.
\newblock Quantum versus classical proofs and advice.
\newblock In {\em Twenty-Second Annual IEEE Conference on Computational
  Complexity (CCC'07)}, pages 115--128. IEEE, 2007.

\bibitem{ME99}
Michele Mosca and Artur Ekert.
\newblock The hidden subgroup problem and eigenvalue estimation on a quantum
  computer.
\newblock In {\em Quantum Computing and Quantum Communications}, pages
  174--188. Springer, 1999.

\bibitem{hallgren07}
Sean Hallgren.
\newblock Polynomial-time quantum algorithms for pell's equation and the
  principal ideal problem.
\newblock {\em Journal of the ACM (JACM)}, 54(1):4, 2007.

\bibitem{EHKS14}
Kirsten Eisentr{\"a}ger, Sean Hallgren, Alexei Kitaev, and Fang Song.
\newblock A quantum algorithm for computing the unit group of an arbitrary
  degree number field.
\newblock In {\em Proceedings of the 46th Annual ACM Symposium on Theory of
  Computing}, pages 293--302. ACM, 2014.

\bibitem{CCDFGS03}
A~.M. Childs, R.~Cleve, E.~Deotto, E.~Farhi, S.~Gutmann, and D~.A. Spielman.
\newblock Exponential algorithmic speedup by a quantum walk.
\newblock In {\em Proceedings of the Thirty-fifth Annual ACM Symposium on
  Theory of Computing}, STOC '03, pages 59--68, New York, NY, USA, 2003. ACM.

\bibitem{Grover96}
Lov~K. Grover.
\newblock A fast quantum mechanical algorithm for database search.
\newblock In {\em Proceedings of the 28$^{\textrm{th}}$ Annual {ACM} Symposium
  on Theory of Computing}, pages 212--219, New York, 22--24 May 1996.

\bibitem{HHL09}
Aram~W Harrow, Avinatan Hassidim, and Seth Lloyd.
\newblock Quantum algorithm for linear systems of equations.
\newblock {\em Physical review letters}, 103(15):150502, 2009.

\bibitem{CKS15}
Andrew~M Childs, Robin Kothari, and Rolando~D Somma.
\newblock Quantum linear systems algorithm with exponentially improved
  dependence on precision.
\newblock {\em arXiv preprint arXiv:1511.02306}, 2015.

\bibitem{kuperberg05}
Greg Kuperberg.
\newblock A subexponential-time quantum algorithm for the dihedral hidden
  subgroup problem.
\newblock {\em SIAM Journal on Computing}, 35(1):170--188, 2005.

\bibitem{FKW02}
Michael~H Freedman, Alexei Kitaev, and Zhenghan Wang.
\newblock Simulation of topological field theories by quantum computers.
\newblock {\em Communications in Mathematical Physics}, 227(3):587--603, 2002.

\bibitem{BFLW05b}
M~Bordewich, M~Freedman, L~Lov{\'a}sz, and D~Welsh.
\newblock Approximate counting and quantum computation.
\newblock {\em Combinatorics, Probability and Computing}, 14(5-6):737--754,
  2005.

\bibitem{AJL09}
Dorit Aharonov, Vaughan Jones, and Zeph Landau.
\newblock A polynomial quantum algorithm for approximating the jones
  polynomial.
\newblock {\em Algorithmica}, 55(3):395--421, 2009.

\bibitem{Rosen05}
Kenneth~H Rosen.
\newblock {\em Elementary Number Theory and its Applications, 2005}.
\newblock Addison Wesley.

\bibitem{KST+94}
Johannes K{\"o}bler, Uwe Sch{\"o}ning, and Jacobo Tor{\'a}n.
\newblock {\em The graph isomorphism problem: its structural complexity}.
\newblock Birkhauser Verlag, 1994.

\bibitem{Babai16}
L\'{a}szl\'{o} Babai.
\newblock Graph isomorphism in quasipolynomial time [extended abstract].
\newblock In {\em Proceedings of the 48th Annual ACM SIGACT Symposium on Theory
  of Computing}, STOC 2016, pages 684--697, New York, NY, USA, 2016. ACM.

\bibitem{Yurke86}
B~Yurke.
\newblock Input states for enhancement of fermion interferometer sensitivity.
\newblock {\em Physical review letters}, 56(15):1515, 1986.

\bibitem{Sanders89}
Barry~C Sanders.
\newblock Quantum dynamics of the nonlinear rotator and the effects of
  continual spin measurement.
\newblock {\em Physical Review A}, 40(5):2417, 1989.

\bibitem{LKD02}
Hwang Lee, Pieter Kok, and Jonathan~P Dowling.
\newblock A quantum rosetta stone for interferometry.
\newblock {\em Journal of Modern Optics}, 49(14-15):2325--2338, 2002.

\bibitem{Fleming73}
Gordon~N Fleming.
\newblock A unitarity bound on the evolution of nonstationary states.
\newblock {\em Il Nuovo Cimento A (1971-1996)}, 16(2):232--240, 1973.

\bibitem{Bhattacharyya83}
Kamal Bhattacharyya.
\newblock Quantum decay and the mandelstam-tamm-energy inequality.
\newblock {\em Journal of Physics A: Mathematical and General}, 16(13):2993,
  1983.

\bibitem{BACS07}
Dominic~W Berry, Graeme Ahokas, Richard Cleve, and Barry~C Sanders.
\newblock Efficient quantum algorithms for simulating sparse hamiltonians.
\newblock {\em Communications in Mathematical Physics}, 270(2):359--371, 2007.

\bibitem{ChildsThesis}
Andrew~Macgregor Childs.
\newblock {\em Quantum information processing in continuous time}.
\newblock PhD thesis, Massachusetts Institute of Technology, 2004.

\bibitem{CK11}
Andrew~M Childs and Robin Kothari.
\newblock Simulating sparse hamiltonians with star decompositions.
\newblock In {\em Theory of Quantum Computation, Communication, and
  Cryptography}, pages 94--103. Springer, 2011.

\bibitem{CW12}
Andrew~M Childs and Nathan Wiebe.
\newblock Hamiltonian simulation using linear combinations of unitary
  operations.
\newblock {\em arXiv preprint arXiv:1202.5822}, 2012.

\bibitem{BBCMW98}
Robert Beals, Harry Buhrman, Richard Cleve, Michele Mosca, and Roland de~Wolf.
\newblock Quantum lower bounds by polynomials.
\newblock In {\em Proceedings of the 39th IEEE Conference on Foundations of
  Computer Science}, pages 352--361, 1998.

\bibitem{Maldacena99}
Juan Maldacena.
\newblock The large-n limit of superconformal field theories and supergravity.
\newblock {\em International journal of theoretical physics}, 38(4):1113--1133,
  1999.

\bibitem{FW62}
Robert~W Fuller and John~A Wheeler.
\newblock Causality and multiply connected space-time.
\newblock {\em Physical Review}, 128(2):919, 1962.

\bibitem{Aaronson16}
Scott Aaronson.
\newblock The complexity of quantum states and transformations: From quantum
  money to black holes.
\newblock {\em arXiv preprint arXiv:1607.05256}, 2016.

\bibitem{NH14}
Rahul Nandkishore and David~A Huse.
\newblock Many body localization and thermalization in quantum statistical
  mechanics.
\newblock {\em arXiv preprint arXiv:1404.0686}, 2014.

\bibitem{HGAKR15}
David~A Herrera-Mart{\'\i}, Tuvia Gefen, Dorit Aharonov, Nadav Katz, and Alex
  Retzker.
\newblock Quantum error-correction-enhanced magnetometer overcoming the limit
  imposed by relaxation.
\newblock {\em Physical review letters}, 115(20):200501, 2015.

\bibitem{KLSL14}
Eric~M Kessler, Igor Lovchinsky, Alexander~O Sushkov, and Mikhail~D Lukin.
\newblock Quantum error correction for metrology.
\newblock {\em Physical review letters}, 112(15):150802, 2014.

\end{thebibliography}
\bibliographystyle{unsrt}

\end{document}